\documentclass[11pt]{article}
\usepackage{enumerate}
\usepackage{comment}
\usepackage{url}
\usepackage{subfig}
\usepackage{amsmath}
\usepackage{amsthm}
\usepackage{amssymb}
\usepackage{wrapfig}
\usepackage{graphicx}
\usepackage{paralist}
\usepackage[active]{srcltx}
\usepackage{algorithm}
\usepackage{fullpage}
\usepackage[english]{babel}
\newtheorem{theorem}{Theorem}

\newtheorem{coro}[theorem]{Corollary}
\newtheorem{lemma}[theorem]{Lemma}
\newtheorem{proposition}[theorem]{Proposition}

\newtheorem{defn}[theorem]{Definition}

\newcommand{\route}{\text{\sc{ipp}}}
\newcommand{\IPP}{\route}

\newcommand{\flow}{\text{\emph{flow}}}

\newcommand{\opt}{\text{\textsc{opt}}}

\newcommand{\pmax}{p_{\max}}
\newcommand{\pmaxst}{p_{\max}^{st}}

\newcommand{\eqdf}{\triangleq}

\newcommand{\alg}{\text{\sc{alg}}}
\newenvironment{proof sketch}[1]{\noindent {\emph{Proof sketch of #1:}}}{\hfill \qed}
\newcommand{\cut}{\textit{cut}}
\newcommand{\diam}{\textit{diam}}
\newcommand{\dist}{\textit{dist}}

\newcommand{\far}{\text{\emph{Far}}}
\newcommand{\near}{\text{\emph{Near}}}
\newcommand{\algf}{\alg_{\far^+}}
\newcommand{\algn}{\alg_{\near}}
\newcommand{\RIPP}{\IPP(\far^+ \mid \pmax)}
\newcommand{\RIPPt}{\IPP^{\lambda}}
\newcommand{\Rinj}{\RIPPt_{1/4}}
\newcommand{\algfar}{\algf}
\newcommand{\hl}{\tau}
\newcommand{\vl}{Q}

\begin{document} 
\title{Online Packet-Routing in Grids with Bounded Buffers \thanks{Preliminary
    versions of this manuscript appeared in the proceedings of ICALP
    2010~\cite{DBLP:conf/icalp/EvenM10} and SPAA 2011~\cite{DBLP:conf/spaa/EvenM11}.}
}

\author{Guy Even
\thanks{School of Electrical Engineering, Tel-Aviv Univ., Tel-Aviv 69978, Israel. ({\tt guy@eng.tau.ac.il}).}
\and
Moti Medina
\thanks{School of Electrical Engineering, Tel-Aviv Univ., Tel-Aviv 69978, Israel. ({\tt medinamo@eng.tau.ac.il}).
Partially funded
by the Israeli Ministry of Science and Technology.}}

\maketitle
\begin{abstract}

We present deterministic and randomized algorithms for the problem
  of online packet routing in grids in the competitive network throughput
  model~\cite{AKOR}. In this model the network  has nodes with bounded buffers and bounded
  link capacities. The goal in this model is to maximize the throughput, i.e., the
  number of delivered packets.

  Our deterministic algorithm is the first online
  algorithm with an $O\left(\log^{O(1)}(n)\right)$
  competitive ratio for uni-directional grids (where $n$
  denotes the size of the network).  The deterministic
  online algorithm is centralized and handles packets
  with deadlines.  This algorithm is applicable to
  various ranges of values of buffer sizes and
  communication link capacities. In particular, it holds
  for buffer size and communication link capacity in the
  range $[3 \ldots \log n]$.

Our randomized algorithm achieves an expected competitive ratio of $O(\log n)$ for
the uni-directional line. This algorithm is applicable to a wide range of buffer
  sizes and communication link capacities. In particular, it holds also for unit size
  buffers and unit capacity links.  This algorithm improves the best previous
  $O(\log^2 n)$-competitive ratio of Azar and Zachut~\cite{AZ}.
\end{abstract}

\thispagestyle{empty}
\section{Introduction}
Large scale communication networks partition messages into packets so that high
bandwidth links can support multiple sessions simultaneously. Packet routing is used
by the Internet as well as telephony networks and cellular networks. Thus, the
development of algorithms that can route packets between different pairs of nodes is
a fundamental problem in networks. In a typical setting, requests for routing packets
arrive over time, thus calling for the development of online packet routing
algorithms. The holy grail of packet routing is to develop online distributed
algorithms whose performance is competitive with respect to multiple criteria, such
as: throughput (i.e., deliver as many packets as possible), delay (i.e., guarantee
arrival of packets on time), stability (e.g., constant rate, avoid buffer overflow) ,
fairness (i.e., fair sharing of resources among users), etc.  From a theoretical
point of view, there is still a huge gap between known lower bounds and upper bounds
for packet routing even in the simple setting of directed paths and centralized
algorithms.

We study the ``Competitive Network Throughput Model'' introduced by~\cite{AKOR} for
dynamic routing on networks with bounded buffers.  The goal is to route packets
(i.e., constant length formatted data) in a network of $n$ nodes.  Nodes in this
model are switches with local memories called buffers.  An incoming packet is either
forwarded to a neighbor switch, stored in the buffer, or erased. The resources of a
packet network are specified by two parameters: $c$ - the capacity of links and $B$ -
the size of buffers.  The capacity of a link is an upper bound on the number of
packets that can be transmitted in one time step along the link.  The buffer size is
the maximum number of packets that can be stored in a node.

\subsection{Previous Work}
Algorithms for dynamic routing on networks with bounded
buffers have been studied both in theory and in practice.
The networks we study are uni-directional grids of $d$
dimensions. Such $2$-dimensional grids with or without
buffers serve as crossbars in networks
(see~\cite{ARSU,AKRR,T} for many references from the
networking community). Thus, even centralized algorithms
for this task are of interest since they can be used to
control a crossbar.

\paragraph{Online Algorithms for Uni-directional Lines.}
Our work on uni-directional line networks is based on a
sequence of papers starting with~\cite{AKOR}.
In~\cite{AKOR}, a lower bound of $\Omega(\sqrt{n})$ was
proved for the greedy algorithm on uni-directional lines if
the buffer size $B$ is at least two. For the case $B=1$ (in
a slightly different model), an $\Omega(n)$ lower bound for
any deterministic algorithm was proved by~\cite{AZ,AKK}.
Both~\cite{AZ} and~\cite{AKK} developed, among other
things, online randomized centralized algorithms for
uni-directional lines with $B>1$. In~\cite{AKK} an
$O(\log^3 n)$-competitive randomized centralized algorithm
was presented for buffer size $B$ at least $2$. For the
case $B\geq 2$, ~\cite{AKK} proved that nearest-to-go is
$\tilde{O}(\sqrt{n})$-competitive. For the case $B=1$,
~\cite{AKK} presented a randomized
$\tilde{O}(\sqrt{n})$-competitive distributed algorithm.
(This algorithm also applies to rooted trees when the
packet destinations are the root.) In~\cite{AZ}, an
$O(\log^2 n)$-competitive randomized algorithm was
presented for the case $B\geq 2$. (This algorithm also
applies to rings and trees.)

\paragraph{Online Algorithms for Uni-directional Grids.}
Angelov et al.~\cite{AKK} showed that the competitive ratio
of greedy algorithms in uni-directional $2$-dimensional
grids is $\Omega(\sqrt{n})$ and that nearest-to-go policy
achieves a competitive ratio of $\tilde{\Theta}(n^{2/3})$.

\paragraph{Other Related Results.}
Kleinberg and Tardos~\cite{KT} studied the disjoint path problem in undirected planar
graphs (see~\cite{KT} for a formal description of the family of graphs for which
their results hold).  They presented constant approximation randomized algorithm for
this problem as well as an online algorithm with logarithmic competitive ratio.

Leighton et al.~\cite{leighton1994packet} and subsequent
works~\cite{leighton1999fast,RT, srinivasan1997constant} deal with a different model
for packet routing. In this model, there are unbounded input queues and bounded
intermediate buffers. In addition, each packet comes with a path along which it is
sent. The latency of each packet is $O(C+D)$, where $C$ denotes the maximum
congestion and $D$ denotes the length of a longest path.

Offline algorithms for trees and meshes were studied in~\cite{AKRR} . They obtained a
logarithmic approximation ratio for unbounded buffers and a constant approximation
ratio for bufferless networks. Offline packet routing for uni-directional lines was
studied in~\cite{RR}.

\renewcommand{\arraystretch}{1.3}
\begin{table}
\begin{centering}
\begin{tabular}{|c|c|c|c|c|c|}
\hline
Paper & $d$ & Competitive Ratio  & Det. $\backslash$ Rand.& $B$ & Remarks \tabularnewline
\hline
\hline
\cite{AKK}  & $2$ & $\tilde{\Theta}(n^{2/3})$  & det. & $>1$ &  distributed, nearest-to-go,
$1$-bend routing\tabularnewline
\cite{AKK}  & $1$ & $\tilde{O}(\sqrt{n})$  &det.  & $>1$ & distributed, nearest-to-go\tabularnewline
\cite{AKK}  & $1$ & $\tilde{O}(\sqrt{n})$  &rand.  & $=1$ & shared randomness, distributed\tabularnewline
\cite{AKK}  & $1$ & $O(\log^3 n)$  &rand.  & $>1$ &  centralized\tabularnewline
\cite{AZ}   & $1$ & $O(\log^2 n)$ & rand. &  $>1$ &  centralized, FIFO buffers\tabularnewline
\hline
\end{tabular}
\par\end{centering}
\caption{Previous online algorithms for packet routing. The networks are
  uni-directional lines or two dimensional directed grids with
  unit link capacities.}
\label{table:previous work}
\end{table}
\renewcommand{\arraystretch}{1}

\subsection{Our Results}
We present online algorithms for packet routing in
$d$-dimensional uni-directional grids (for $d=O(1)$) as
follows.

\paragraph{Deterministic Online algorithm.}
We present a centralized \emph{deterministic} online
algorithm for packet routing in uni-directional grids with
$n$ nodes. Our algorithm achieves a polylogarithmic
competitive ratio for a wide combination of parameters
described below. (The buffer size is denoted by $B$ and the
link capacities are denoted by $c$.) The deterministic
packet-routing algorithm handles requests with deadlines,
allows  preemptions (i.e., packets may be dropped before
they reach their destination), and employs adaptive routing
(i.e., part of the route is computed while the packet is
traveling to its destination).

\begin{enumerate}[(i)]
\item For $B,c\in [3 \ldots \log n]$, the competitive
    ratio of the algorithm is $O(\log ^{d+4} n)$ for
    uni-directional grids of dimension $d$.
\item For $B=0$ and $c\geq 3$, the competitive ratio of
    the algorithm is $O(\log ^{d+2} n)$ for
    uni-directional grids of $d$ dimensions. In the
    trivial case of a uni-directional line (i.e., $d=1$),
    our algorithm is degenerated to the nearest-to-go
    policy~\cite{AKOR} and is optimal.
\item For $B,c \geq \log n$ and $B/c=n^{O(1)}$ the
    algorithm reduces to online integral path
    packing~\cite{BN06,AAP}. The competitive ratio of the
    algorithm is $O(\log n)$ for uni-directional grids,
    independent of the dimension $d$. In this algorithm,
    packets are either rejected or routed but not
    preempted.
\end{enumerate}
In the rest of the paper, we address the algorithm for
uni-directional grids as the `deterministic' algorithm.

\paragraph{A Randomized Algorithm for the One Dimensional Case.}
We present a centralized online \emph{randomized} packet
routing algorithm for maximizing throughput in
uni-directional lines\footnote{We remark that the
randomized
  algorithm can be generalized to $d$-dimensional grids to obtain competitive ratios
  that are $(O(\log n))^d$. In light of similar competitive ratios with the
  deterministic algorithm, we omit the description and analysis of the randomized
  algorithm for $d$-dimensional grids.}. Our algorithm is \emph{nonpreemptive};
rejection is determined upon arrival of a packet. Our algorithm is centralized and
randomized and achieves an $O(\log n)$-competitive ratio. In addition to handling the
case that $B=1$ and $c=1$, our algorithm improves over previous algorithms as
follows:
\begin{enumerate}[(i)]
\item The competitive ratio is $O(\log n)$ compared to the best  previous competitive ratio of $O(\log^2 n)$ by Azar and Zachut~\cite{AZ}.
\item Our algorithm works also for buffers of size $B = 1$ (with no restriction on
  the link capacities).
\item We consider also the parameter $c$ of the capacity of the links (\cite{AZ,AKK} considered only the case $c=1$).
\item The $O(\log n)$ competitive ratio applies for the following combination of parameters:
  (1)~$B\in [1,\log n]$ and $c\geq 1$, or (2)~$\log n\leq
  B/c\leq n^{O(1)}$ .
\end{enumerate}
In the rest of the paper, we address the algorithm for
uni-directional lines as the `randomized' algorithm.

\subsection{Techniques}\label{sec:tech}

\paragraph{Reduction of Packet-Routing to Circuit Switching.}
Packet routing is reduced to a circuit switching problem~\cite{KT,AAP} by applying a
\emph{space-time transformation}~\cite{AAF,ARSU,AZ,RR}.  We extend the space-time
transformation of~\cite{AZ} so that it also supports deadlines.

The reduction of packet routing to circuit switching relies on the ability to bound
the path lengths without losing too much throughput. In~\cite{AZ} a bound on the path
lengths that incurs only a constant fraction loss of throughput is proven for routing
in a uni-directional line. We extend the lemma of~\cite{AZ} to $d$-dimensional grids
and to general values of buffer sizes $B$ and link capacities $c$.

This implies that online packet-routing is reduced to the well studied problem of
online packing of paths~\cite{AAP, BN06}.  Algorithms for online packing of paths
either reject a request or assign a path to a request (i.e., perform call admission).  The
edge capacities of the space-time graph are $B$ and $c$.  If the capacities are
large, i.e., $B,c \geq \log n$, then the online path packing algorithm by Awerbuch
et. al~\cite{AAP} achieves a $\log n$ competitive ratio, where $n$ is number of
vertices of the (original) graph, as required.  In the case where the capacities are
small, i.e., $B,c < \log n$, the algorithm by~\cite{AAP} does not apply, hence we
coalesce groups of nodes by \emph{tiling}~\cite{KT,BL}. This induces a new graph,
called a \emph{sketch graph} in which the capacities are (again) large.  We apply the
online path packing algorithm over the sketch graph, but are left with the problem of
translating paths over the sketch graph to paths over the space-time graph. We refer
to this translation as \emph{detailed routing}. We use the framework of Buchbinder
and Naor~\cite{BN06,BNsurvey} for \emph{online path packing} because it helps us
point out the tradeoffs between the path lengths, the competitive ratios, and the
overloading of edges.

\paragraph{Detailed Routing.}
The path packing algorithm computes a path over the sketch graph, and the algorithm
must translate this sketch path to a detailed path over the space-time graph. The
detailed path traverses the same tiles that are traversed by the sketch path and
bends whenever the sketch path bends.  Detailed routing has been addressed before in
undirected graphs~\cite{KT,BL} as well as in space-time graphs of
the uni-directional line~\cite{RR}.

Detailed routing is not always successful; indeed, we need to bound the fraction of
the requests that are lost during detailed routing.  In the deterministic algorithm,
the detailed routing technique partitions each path in the sketch graph into three
parts, and reserves only a unit of capacity for each part. This is the reason why the
algorithm requires $B,c \geq 3$.  In some parts of the detailed routing, we reduce
the problem of detailed routing to \emph{online interval packing}.  This reduction
uses an online procedure for packing intervals on a line (which is, in fact, a
nearest-to-go routing policy). We apply an online distributed simulation of the
optimal interval packing algorithm~\cite{GLL}. The correctness of this simulation is
based on the ability of the packet-routing algorithm to preempt (i.e., drop) packets.

\paragraph{Classify and Select.}
Requests are categorized as \emph{near} or \emph{far}, and the algorithm randomly chooses to deal
with one category of requests.
The categorization is based on the tiles. A request that can be routed within a tile
is considered near; otherwise it is a far request.

Randomization is also employed to choose a random subset of the requests so as to
further weaken the adversary. We use \emph{random phase shifts} that determine the quadrants
within tiles from which paths may start.

\paragraph{Random Sparsification.}
Requests that are assigned
sketch paths by the online path packing algorithm are randomly sparsified. This
\emph{random sparsification} has two roles: (1)~Reduction of loads of sketch graph edges
incurred by the path packing algorithm to a small constant fraction with high
probability.  (2)~Solving the problem that the source nodes of requests may be
densely packed in an area $A$.  The capacity of the edges that enable routing paths
out of $A$ is proportional to the ``perimeter'' of $A$, while the number of source
nodes in $A$ is proportional to the ``area'' of $A$.  In a $d$ dimensional grid, the
area of a subregion can be as large as the perimeter of the subregion to the power
$d$.  By applying random sparsification, the number of remaining paths whose source
node is in a quadrant of a tile roughly equals the perimeter of the quadrant.

\subsection{Organization}
The formal definition of the problem is stated in Sec.~\ref{sec:problem}.
In Sec.~\ref{sec:prelim}, the reduction of packet-routing to path packing is presented.
In Sec.~\ref{sec:outline}, we outline the steps of the deterministic algorithm.
In Sec.~\ref{sec:alg}, we elaborate on each step of the deterministic algorithm with respect to uni-directional lines and prove that the algorithm is $O(\log^5 n)$-competitive, where $n$ is the number of nodes. %
In Sec.~\ref{sec:generalizations} we present a
generalization of the deterministic algorithm to the
$d$-dimensional case and extensions to special cases, such
as: bufferless grids, and grids with large buffers and
large link capacities. In Sec.~\ref{sec:randalg} we design
and analyze a randomized algorithm for uni-directional
lines. Our randomized algorithm achieves a competitive
ratio of $O(\log n)$.

\section{Problem Definition}\label{sec:problem}
\label{sect:problem}

\subsection{Store-and-Forward Packet Routing Networks}
We consider a synchronous store-and-forward packet
routing network~\cite{AKOR,AKK,AZ}.

Each packet is specified by a $4$-tuple $r_i=(a_i,b_i,t_i,d_i)$, where $a_i\in V$ is the source node of the packet, $b_i\in V$ is the destination node,  $t_i\in \mathbb{N}$ is the time step in which the packet is input to $a_i$, and $d_i$ is the deadline.  Since we consider an online setting, no information is known about a packet $r_i$ before time $t_i$.  Deadlines mean that the algorithm is only credited for delivering packet $r_i$ to its destination $b_i$ before time $d_i$.

The network is a directed graph $G=(V,E)$.
Each edge has a capacity $c$ that specifies the number of packets that can be transmitted along the edge in one time step.  Each node has a local buffer of size $B$ that can store at most $B$ packets.
Each node has a local input through which multiple packets may be input in each time step.  The network operates in a synchronous fashion with a delay of one time step for communication. This means that a single time step is needed for a packet to traverse a single link.

In each time step, a node $v$ considers the packets arriving via the local input, the packets arriving from incoming edges, and the packets stored in the buffer.  Packets destined to node $v$ (i.e., $b_i=v$) are removed from the network (this is considered a success provided that the deadline has not passed, and no further routing of the packet is required). As for the other packets, the node determines which packets are sent along outgoing edges (i.e., forwarded) and which packets are stored in the buffer. The remaining packets are \emph{deleted}.

The literature contains two different models of node functionality. We use the model used by~\cite{ARSU,RR}. The reader is referred to Appendix~\ref{sec:model} for a comparison between two different models of node functionality; this comparison is mostly of interest for the case $B=1$.

We use the following terminology.  A packet is \emph{rejected} if it is locally input to a node and the node deletes it. A packet that is locally input but not rejected is called an \emph{injected} packet. A packet is \emph{preempted} or \emph{dropped} if it was injected and deleted before it reached its destination.

The task of \emph{admission control} is to determine which packets are injected and which are rejected.  An algorithm that drops packets is a \emph{preemptive algorithm}; an algorithm that does not drop packets is called a \emph{non-preemptive algorithm}.

\subsection{Grid Networks}
A two dimensional $\ell_1\times \ell_2$ uni-directional
grid network is a directed graph $G=(V,E)$ defined as
follows (see Fig.~\ref{fig:grid}). The set of vertices is
$V\triangleq[\ell_1]\times [\ell_2]$, where $[\ell]$
denotes the set of integers $\{1,\ldots,\ell\}$.  We denote
the number of vertices by $n$ (i.e., $n=\ell_1\cdot
\ell_2)$.  There are two types of edges: horizontal edges
$(i,j)\rightarrow (i+1,j)$ and vertical edges
$(i,j)\rightarrow (i,j+1)$. For each packet, the source
node $a_i=(a_i(x),a_i(y))$ and the destination node
$b_i=(b_i(x),b_i(y))$ satisfy $a_i \leq b_i$ (i.e.,
$a_i(x)\leq b_i(x)$ and $a_i(y)\leq b_i(y)$).  We refer to
an $\ell_1\times \ell_2$ two dimensional directed grid
network simply as a grid.

A $d$-dimensional grid is defined analogously over a vertex set
$V\triangleq[\ell_1]\times \cdots \times [\ell_d]$. Our analysis applies to the case
that $d$ is a constant.

\paragraph{Capacities and Buffers.}
We assume uniform capacities and buffer sizes. Namely, (i)~all edges in the grid have
the same capacity, denoted by $c$; and (ii)~all nodes have the same buffer size,
denoted by $B$.

  \begin{figure}[H]
    \begin{center}
      \includegraphics[width=0.2\textwidth]{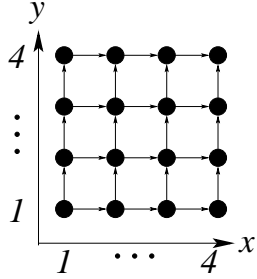}
    \end{center}
    \caption{A $4\times 4$ grid network.}
    \label{fig:grid}
  \end{figure}

\subsection{Online Maximum Throughput in Networks}
The \emph{throughput} of a packet routing algorithm is the number of packets that are delivered to their destination \emph{before their deadline}.
We consider the problem of maximizing the throughput of an online centralized deterministic packet-routing algorithm.

Let $\sigma$ denote an input sequence. Let $\alg$ denote a packet-routing
algorithm. Let $\alg(\sigma)$ denote the subset of requests in $\sigma$ that
are delivered on time by $\alg$. The throughput obtained by $\alg$ on input
$\sigma$ is the size of the set $\alg(\sigma)$, i.e., $|\alg(\sigma)|$. Let
$\opt(\sigma)$ denote the subset of requests in $\sigma$ that are delivered
by an optimal throughput routing. An online deterministic \alg\ is
\emph{$\rho$-competitive} if for every input sequence $\sigma$,
$|\alg(\sigma)| \geq \frac 1\rho \cdot |\opt(\sigma)|$. An online randomized
algorithm is $\rho$-competitive with respect to an oblivious adversary, if
for every input sequence $\sigma$, $\mathbb{E}[|\alg(\sigma)|] \geq \rho \cdot
|\opt(\sigma)|$, where the expected value is over the random choices made by
\alg~\cite{be}.

\subsection{Problem Statement}\label{sec:problem statement}
\paragraph{The Input.}
The online input is a sequence of packet requests $\sigma = \{r_i\}_i$. Each
packet request is  specified by a $4$-tuple $r_i=(a_i,b_i,t_i,d_i)$ over a
grid network $G=(V,E)$. We consider an online setting, namely, the requests
arrive one-by-one, and no information is known about a packet request $r_i$
before its arrival.

\paragraph{The Output.}
In each time step, the packet-routing algorithm decides what each of the packets in
the network should do. This decision can be either reject a new packet,
preempt an existing packet, store a packet in a buffer of the node which the
packet has reached, or forward the packet to a neighboring node.

\paragraph{The Objective.}
The goal is to maximize the number of packets that are
successfully routed (i.e., reach their destination before
the deadline expires).

\section{Reduction of Packet-Routing to Path Packing}
\label{sec:prelim}

\subsection{Space-Time Transformation}
\label{sec:spacetime}

A\emph{ space-time transformation} is a method to map traffic in a directed graph over time into a directed acyclic graph~\cite{AAF,ARSU,AZ,RR}. Consider a directed graph $G=(V,E)$ with edge capacities $c$ and buffer size $B$.
The space-time transformation of $G$ is the acyclic directed infinite graph $G^{st}=(V^{st},E^{st})$ with edge capacities $c^{st}(e)$, where:
\begin{inparaenum}[(i)]
\item $V^{st} \triangleq V\times \mathbb{N}$.
\item $E^{st}\triangleq E_0\cup E_1$ where $E_0\triangleq \{
  (u,t)\rightarrow(v,t+1)\::\: (u,v)\in E~,~t\in\mathbb{N}\}$ and $E_1
  \triangleq \{ (u,t)\rightarrow (u,t+1) \::\: u\in V, t\in
  \mathbb{N}\}$.
\item The capacity of all edges in $E_0$ is $c$, and all edges in $E_1$ have capacity $B$. Note that the space-time graph corresponding to a $d$-dimensional grid is a $(d+1)$-dimensional grid.
\end{inparaenum}
Figure~\ref{fig:st} depicts the space-time transformation in the one dimensional case.

\paragraph{Adding Sink Nodes.}
Following~\cite{AZ}, we add sink nodes to define a specific destination node for each request. For every vertex $v$ in the line, we define a sink node
$\hat v$ (see Figure~\ref{fig:stsink}).
A \emph{copy of a vertex} $v \in V$ in the space-time graph $G^{st}$ is a space-time vertex $(v,t)\in V^{st}$ for some $t$.
We add an incoming edge of infinite capacity to the sink node $\hat v$ from each tile $s$ that contains a copy $(v,t)$ of $v$.

\begin{figure}[h]%
  \centering
\includegraphics[width=0.5\textwidth]{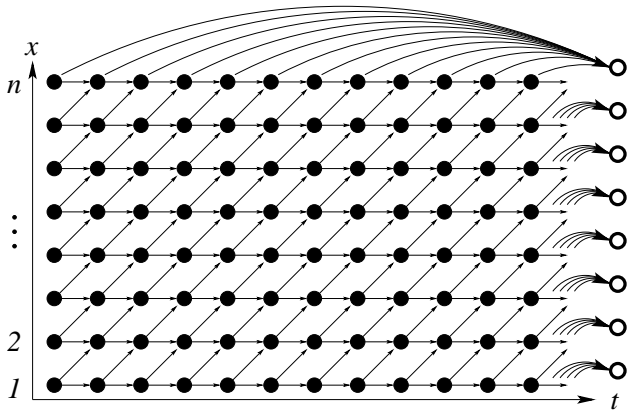}
\caption{The space-time graph $G^{st}$ with the new sink nodes (shown on the rightmost column).
}
\label{fig:stsink}
\end{figure}

\subsection{Untilting}
A standard drawing%
    \footnote{For $d=2$, the $G^{st}$ has a
    $3$-dimensional standard drawing in which: (i)~ a node $(i,j,t)\in
    V^{st}$ is mapped to the point $(i,j,t)$, and (ii)~edges are mapped
    to straight segments between their endpoints.}  %
of the space-time graph of a grid is a lattice generated by non-orthogonal vectors. This drawing is hard to depict and deal with, hence we apply a transformation called untilting defined as follows (see~\cite{RR} for untilting in two dimensions).

We rectify the drawing of the space-time graph of a grid by applying
an automorphism $q:\mathbb{Z}^{d+1} \rightarrow \mathbb{Z}^{d+1}$ defined by
$q(x_1,\ldots,x_d,t)\triangleq (x_1,\ldots,x_d,t-\sum_{i=1}^{d}x_i)$. We refer to this
transformation as \emph{untilting}. The sole purpose of applying
untilting is to obtain a drawing of the space-time graph of a grid in which the edges are axis parallel. Such an axis parallel drawing
simplifies the definition of tiles.
Note that the image of some of the vertices in $G^{st}$ is outside the positive quadrant.
Figure~\ref{fig:stuntilt} depicts the untilted space-time graph in the one dimensional case. (e.g., the node $(2,1)$ is mapped to $(2,-1)$.)

\subsection{Tiling}\label{sect:tiling}
The term \emph{tiling} refers to a partitioning of the nodes of the
space-time graph $G^{st}$ into finite sets with identical geometric
``shape''.

Tiling is obtained by a partitioning of $\mathbb{Z}^{d+1}$ by disjoint
$(d+1)$-dimensional cubes with side-length $k$.  (For the sake of simplicity
$\mathbb{Z}^{d+1}$ is partitioned to cubes. One can save a logarithmic factor in the
competitive ratio by a partitioning to boxes with unequal side
length. See Section~\ref{sect:prelimline} for an example of such a partitioning.)

A tile $s$ is a maximal subset of $V^{st}$ such that its image $q(s)$ (after untilting) is contained in a cube.
Formally, given a cube side-length $k$, a tile is defined by its \emph{lower corner}
$p\in \mathbb{Z}^{d+1}$, where the coordinates of $p$ are integral multiples of $k$.
The lower corner $p$ defines the tile $s_p \eqdf \{v\in V^{st} : p \leq q(v) <
p+k\cdot \vec{1}\}$, where $\vec{1}$ is the all ones vector. Note that some of the
tiles in $V^{st}$ are \emph{partial}, namely contain less than $k^d$ vertices (see
Figures~\ref{fig:dtile},~\ref{fig:dtiletilt}). In this case, we augment partial tiles
by dummy vertices so that they are complete.
Note that a dummy vertex is never an internal vertex in a path between non-dummy
vertices, and hence, this augmentation has no effect on routing.
\subsection{The Sketch Graph}\label{sect:sketchgraph}
The sketch graph is the graph obtained from the space-time graph after coalescing
each tile into a single node (sink nodes remain unchanged).  There is a directed edge
$(s_1,s_2)$ between two tiles $s_1,s_2$ in the sketch graph if there is a directed
edge $(\alpha,\beta)\in E^{st}$ such that $\alpha\in s_1$ and $\beta\in s_2$. The
capacity $c(s_1,s_2)$ of an edge $(s_1,s_2)$ in the sketch graph is simply the sum of
the capacities of the edges in $G^{st}$ from vertices in $s_1$ to vertices in $s_2$
(i.e., the capacity of a vertical edge between two tiles $c\cdot \hl$ and the capacity of a
horizontal edge is $B\cdot \vl$).  Figure~\ref{fig:dsketch} depicts an untilted sketch graph
of a space-time graph of a one dimensional grid.

The sketch graph also has node capacities for nodes that correspond to tiles (i.e.,
not sinks). The capacity of every node that corresponds to a tile is $c(s)=2\cdot
k^2\cdot (B+c)$.
\paragraph{Notation.}
We denote the sketch graph by $S=(V(S),E(S))$.
We abuse notation and often refer to the nodes of $S$ (that are not sinks) as tiles.

\begin{figure}%
  \centering
  \subfloat[Space-time graph $G^{st}$]{\label{fig:st}
\includegraphics[width=0.5\textwidth]{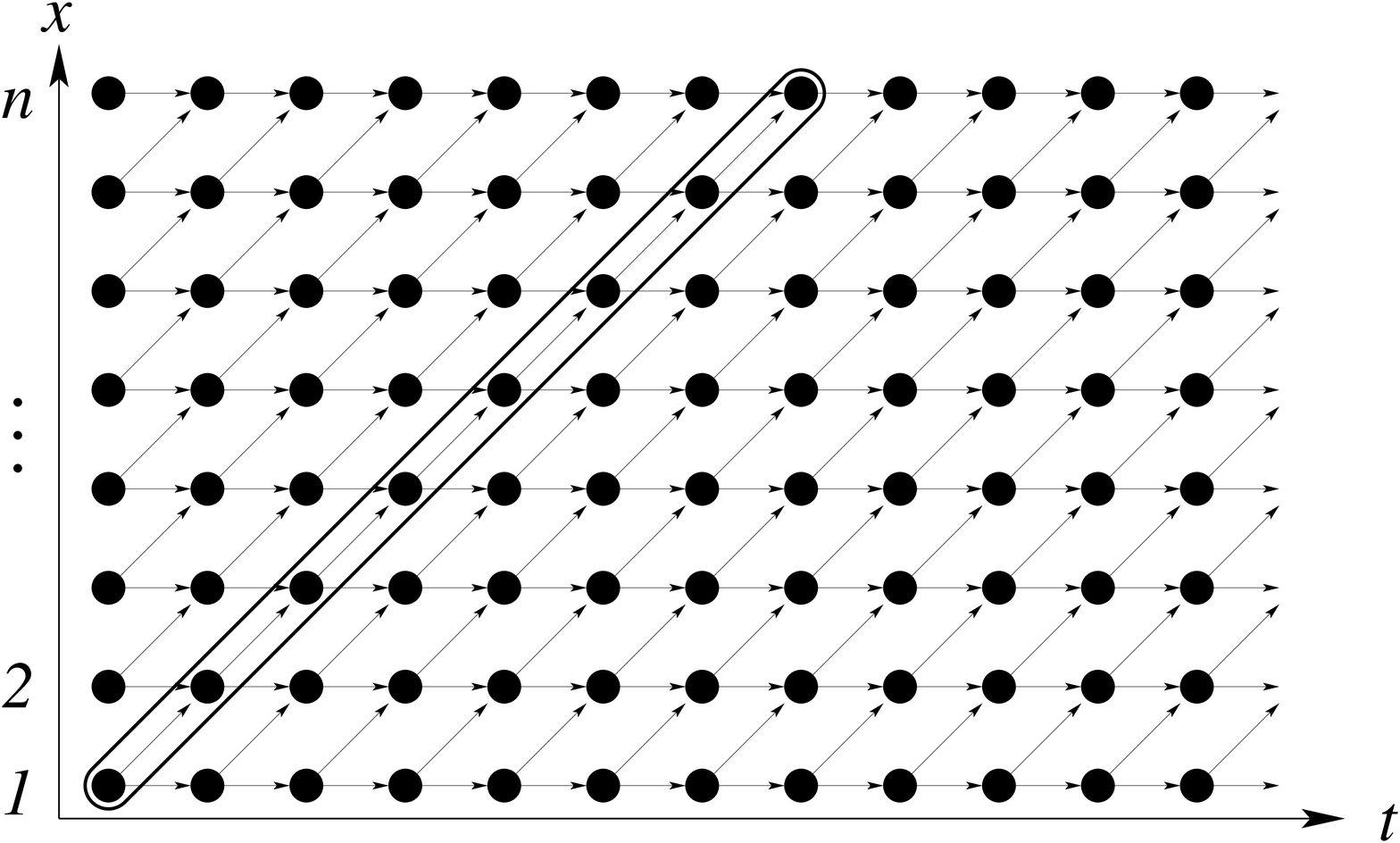}
} \qquad
  \subfloat[Untilted space-time graph $q(G^{st})$]{\label{fig:stuntilt}
\includegraphics[width=0.75\textwidth]{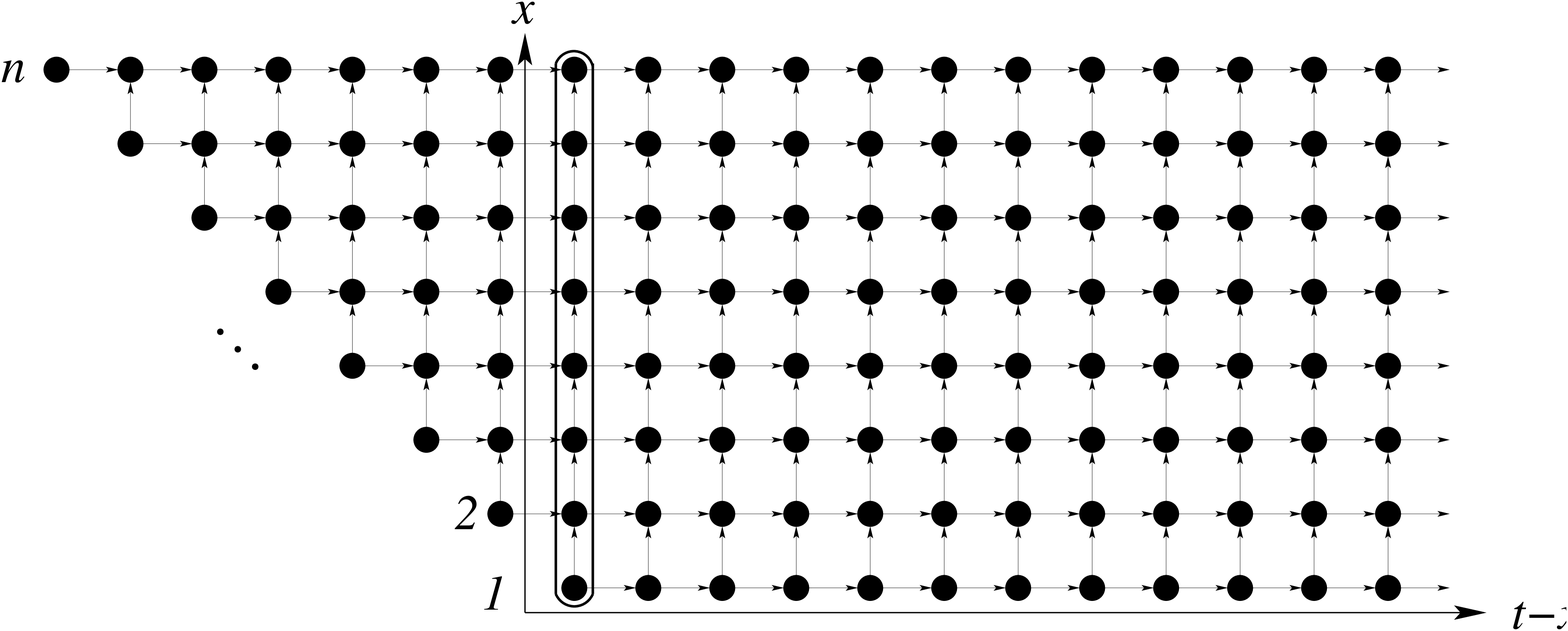}
}\\%
  \subfloat[Tiles in $G^{st}$]{\label{fig:dtiletilt}
\includegraphics[width=0.4\textwidth]{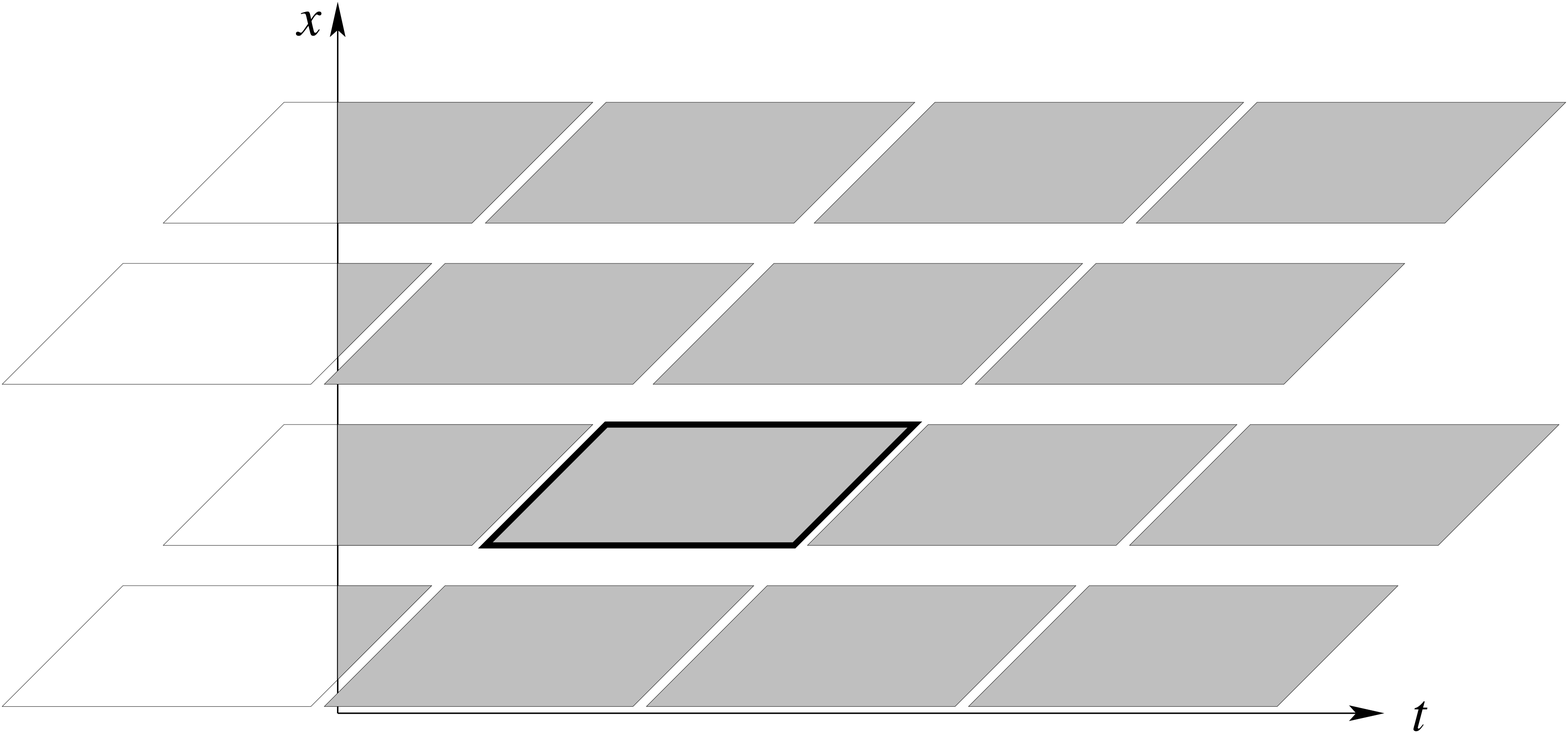}
} \quad
\subfloat[Tiles in $q(G^{st})$]{\label{fig:dtile}
\includegraphics[width=0.4\textwidth]{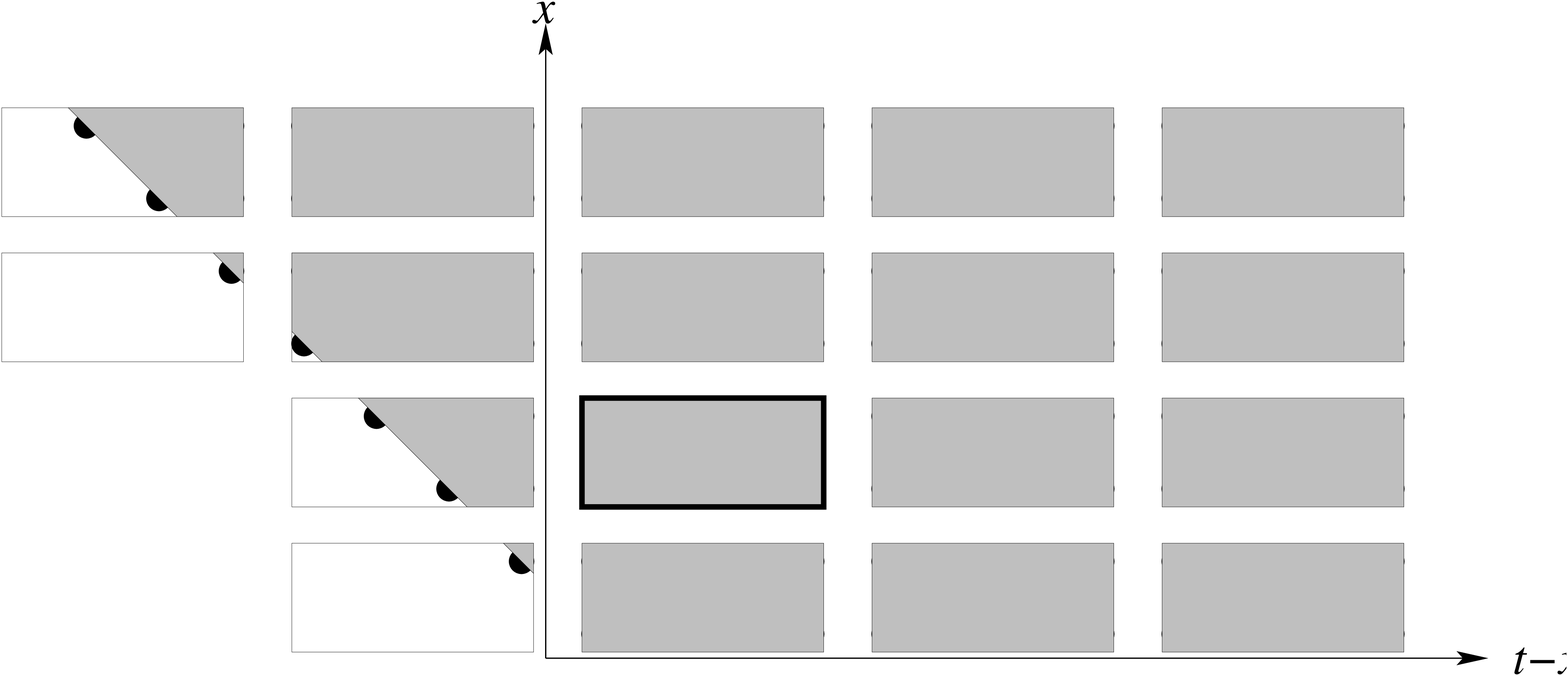}
}\quad
  \subfloat[Sketch graph $S$]{\label{fig:dsketch}
\includegraphics[width=0.4\textwidth]{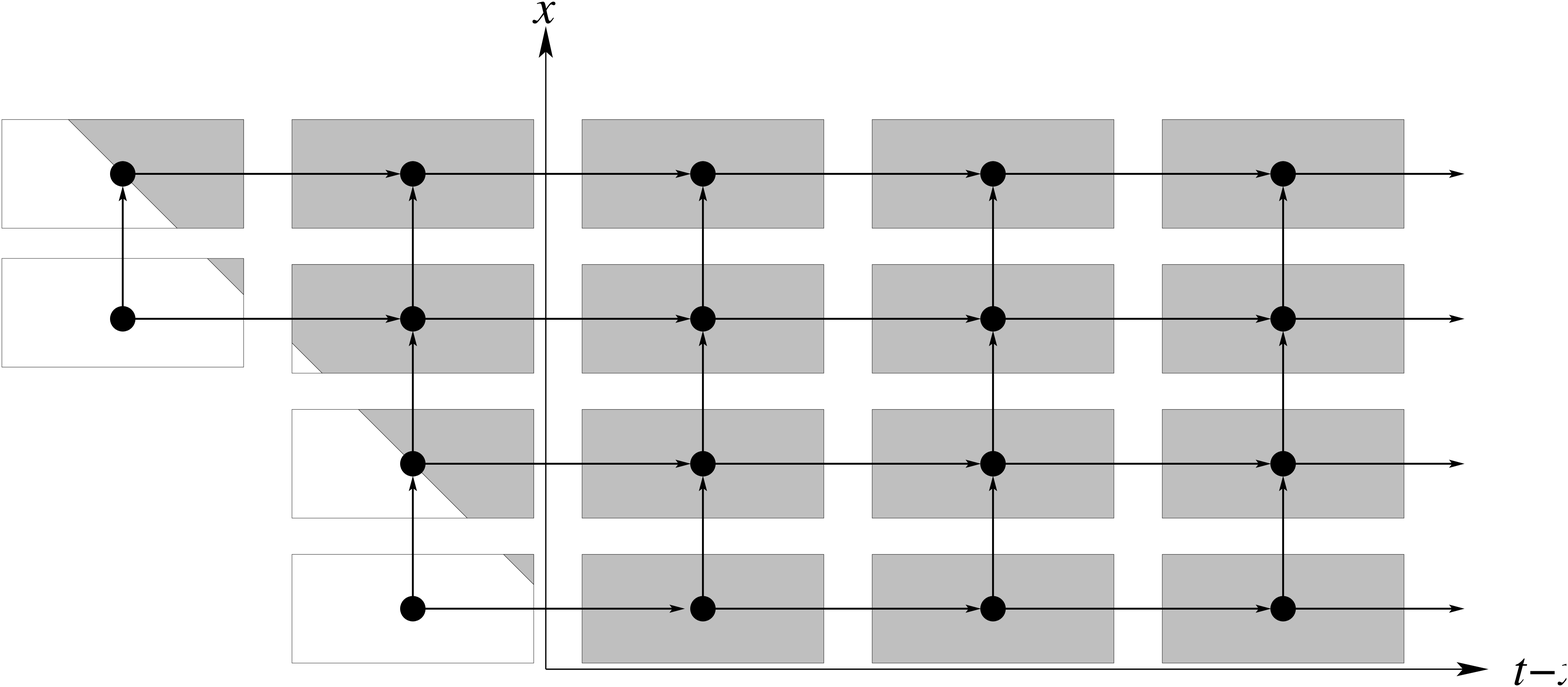}
}
\caption{\footnotesize(a) The tilted space-time graph $G^{st}$. The  horizontal axis is the (infinite) time axis and the vertical axis is the (finite) node axis.
(b) The untilted space-time graph $G^{st}$. The encapsulated path in (a) corresponds to the encapsulated path in (b).
Diagonal edges depict edges in $E_0$. Edges in $E_1$ are depicted by horizontal edges.
(c) The corresponding tiling of the (tilted) space-time graph $G^{st}$. The bolded parallelogram in (c) corresponds to the bolded rectangle in (d).
(d) Tiling of the untilted space-time graph $G^{st}$ by $2\times 4$ rectangles.
(e) The sketch graph over the tiles $S$.
}
\end{figure}

\subsection{Online Packing of Paths}
\label{sect:routing}

A reduction of packet routing to packing of paths is presented in Section~\ref{sec:reduce}. We briefly overview the topic of online packing of paths.

Consider a graph $G=(V,E)$ with edge capacities $c(e)$. Edges have soft capacity constraints (i.e., the capacity constraint may be violated, and one goal is to minimize the
violation).
The adversary introduces a sequence of connection requests
$\{r_i\}_i$, where each request is a source-destination pair
$(a_i,b_i)$. The online packing algorithm must either return a path
$p_i$ from $a_i$ to $b_i$ or reject the request.

Consider a sequence $R=\{r_i\}_{i\in I}$ of requests. A sequence
$P=\{p_i\}_{i \in J}$ is a (partial) \emph{routing} with respect to $R$ if
$J\subseteq I$ and each path $p_i$ connects the
source-destination pair $r_i$.  The \emph{load} of an edge $e$ induced
by a routing $P$ is the ratio $|\{p_j \in P: e\in p_j\}|/c(e)$.
A routing $P$ with respect to $R$ is called a \emph{$\beta$-packing} (or \emph{$\beta$-feasible}) if the
load of each edge is at most $\beta$. The \emph{throughput} of a packing
$P=\{p_i\}_{i\in J}$ is simply $|J|$.

An online path packing algorithm is \emph{$(\alpha,\beta)$-competitive} if it computes a $\beta$-packing $P$ whose throughput is at least $1/\alpha$ times the maximum   throughput over all $1$-packings.

If each request is served by a single path, then the routing is
\emph{nonsplittable}.

A \emph{fractional} packing is a multi-commodity flow. Each
demand can be (partly) served by a combination of fractions
of flows along paths. A sequence $P_f=\{P_i\}_{i \in I}$ is
a \emph{fractional (splittable) routing} with respect to
$R$ if each path $p_i \in P_i$ connects the
source-destination pair $r_i$, and the total flow allocated
by paths in $P_i$ is at most one. The \emph{throughput} of
a fractional splittable path packing $P_f=\{P_i\}_{i\in I}$
is the sum of the allocated flows along every path in
$P_f$. An optimal offline fractional packing can be
computed by solving a linear program. Obviously, the
throughput of an optimal fractional packing is an upper
bound on the throughput of an optimal integral packing.

The proof of the following theorem appears in Appendix~\ref{sect:routealg}.
The proof is based on techniques from~\cite{AAP, BN06}.
We refer to the online algorithm for online integral path packing by $\route$.

\begin{theorem}\label{thm:IPP}
  Consider an infinite graph with edge capacities such that $\min_{e}
  c(e) \geq 1$.  Consider an online path packing problem in which a
  path is legal if it contains at most $\pmax$ edges.
  Assume that there is an oracle, that given edge weights and a
  connection request, finds a
  lightest legal path from the source to the destination. Then, there
  exists a $(2,\log(1+ 3\cdot \pmax))$-competitive online
  integral path packing algorithm.  Moreover, the throughput is at least
  $1/2$ times the maximum throughput over all fractional packings.
\end{theorem}

\subsection{Polynomial Path Lengths}
\paragraph{Notation.}
Consider a directed graph $G=(V,E)$ with edge capacities $c(e)$ and
buffer size $B$ in each vertex.  Let $G^{st}$ denote the space-time
graph of $G$ (see Section~\ref{sec:spacetime}).  Let $c_{\min} =
\min\{c(e) \mid e\in E\}$.  Let
$\dist_G(u,v)$ denote the length of a shortest path from $u$ to $v$ in
$G$. Let $\diam(G)$ denote the diameter of $G$ defined as follows
\begin{align*}
  \diam(G) &\eqdf \max \{ \dist_{G}(u,v) \mid \text{there is a path from $u$ to $v$ in $G$}\}.
\end{align*}

Consider a sequence $R=\{r_i\}_i$ of routing requests (without deadlines) over $G^{st}$, i.e., each request is a three-tuple $r_i = (a_i,b_i,t_i)$ that requires a path from $(a_i,t_i)$ to a copy of $b_i$ in $G^{st}$, that is, $(b_i,t)$ for $t \geq t_i$.

Let $\opt_f(R)$ denote an optimal fractional path packing in $G^{st}$ with respect to $R=\{r_i\}_i$.
Let $\opt_f(R \mid \pmax)$ denote an optimal fractional path packing in $G^{st}$ with respect to $R=\{r_i\}_i$ under the constraint that each request is routed along a path of length at most $\pmax$.
Let $|g|$ denote the throughput of a fractional path packing $g$.

The following lemma shows that bounding path lengths (in a fractional
path packing problem over a space-time graph) by a polynomial
decreases the throughput only by a constant factor.  The
lemma is an extension of a similar lemma from~\cite{AZ}.
The proof of Lemma~\ref{lemma:nB} appears in Appendix~\ref{sec:proofnB}.

\begin{lemma}\label{lemma:nB}
  Let $\alpha\eqdf\frac{c_{\min}}{2\cdot
    (\sum_{e\in E} c(e)+n\cdot B)}$, $\nu\eqdf1/\alpha$, and $\pmax\geq
  (\nu+2)\cdot \diam(G)$.  Then, $$|\opt_f(R \mid \pmax)| \geq \frac
  1{2} \cdot \left(1-\frac 1e\right) \cdot |\opt_f(R)|\:.$$
\end{lemma}

\subsubsection{Remarks}
\begin{inparaenum}[(1)]
\item If $G^{st}$ is the space-time graph of a
    uni-directional line, then we set  $$\pmax\triangleq
    2n  \cdot \left(1+n \cdot
    \left(\frac{B}{c}+1\right)\right)\:.$$
\item If $G^{st}$ is the space-time graph of a
    $d$-dimensional uni-directional grid, then we set
    $$\alpha\eqdf\frac{c_{\min}}{2\cdot (\sum_{e\in E}
    c(e)+n\cdot B)} = \frac{1}{2n\cdot (d +B/c)}\:,$$
    and $$\pmax\triangleq 2 \cdot \diam(G)\cdot \left(1+n\cdot \left(\frac{B}{c} + d \right)\right) \:.$$
\item
 A trivial lower bound on the path lengths is $\Omega(B/c)$ if we want to be
  able to route a constant fraction of the optimal throughput. Indeed, if $B$ packets
  are injected simultaneously to the same node in a line, then at most $c$ packets
  can be forwarded in each step. Hence $\Omega(B/c)$ steps are required to forward a
  constant fraction of the packets. This justifies the term $B/c$ in the definition
  of the maximum path length (see Lemmas~\ref{lemma:nB} and \ref{lemma:nBline}).
\end{inparaenum}

\section{Outline of the Deterministic Algorithm}\label{sec:outline}
The listing of the deterministic framework appears in
Algorithm~\ref{alg:algDet}.  Upon arrival of a request $r_i$, the
algorithm reduces the packet request to an online integral path
packing over the sketch graph with bounded paths. The algorithm then
executes the online algorithm for online integral path packing (\IPP)
with respect to this path request. If the path request is rejected by
the \IPP\ algorithm, then the algorithm rejects $r_i$.  Otherwise, let
$\hat p_i$ denote the sketch path assigned to the request $r_i$. The
algorithm injects the request $r_i$ with its sketch path $\hat p_i$
and performs detailed routing in the space-time graph $G^{st}$.
Detailed routing in $G^{st}$ may fail (see
Section~\ref{sec:detailed}). In case of failure, the algorithm
preempts $r_i$.

To simplify the description, we begin in, Sec.~\ref{sec:alg}, by
presenting a detailed description and proof for the one-dimensional
case.  The required modifications for higher dimensions are described
in Sec.~\ref{sec:algd}.  We also assume that there are no deadlines
(i.e., $d_i = \infty$), hence each packet is specified by a $3$-tuple
$r_i=(a_i,b_i,t_i)$; we reintroduce deadlines in
Section~\ref{sec:d_i}.

\begin{algorithm}[H]
    \textbf{Upon arrival} of a packet request $r_i = (a_i,b_i,t_i)$, for $i\geq 1$ (if $r_i$ is rejected or preempted in any step, then the algorithm does not continue with the next steps), the algorithm proceeds as follows:
        \begin{enumerate}
            \item \label{item:reduce}Reduce $r_i = (a_i,b_i,t_i)$ to a path request $\hat r_i$ in the $\{1,2,\infty\}$-sketch graph $\hat S$ as follows:
                \begin{enumerate}
                  \item The source of the path request $\hat r_i$ is the half  tile $s_{in}$, where the tile $s$ contains the vertex $(a_i,t_i)$.
                  \item The destination of the path request $\hat r_i$ is simply the sink $\hat b_i$.
                \end{enumerate}
            \item Execute the \IPP\ algorithm over $\hat S$ with respect to the reduced path request $\hat r_i$.
             \begin{enumerate}
               \item If the \IPP\ algorithm rejects the $\hat r_i$ then \textbf{reject} $r_i$.
               \item Else, let $\hat p_i$ denote the path output by \IPP, i.e., the sketch path assigned to $\hat r_i$.
             \end{enumerate}
            \item \label{item:I} \textbf{Inject} the request $r_i$ (the request ``includes'' its sketch path  $\hat p_i$) and perform detailed routing in the space-time graph $G^{st}$.
                Detailed routing proceeds by processing the \emph{first segment} of $\hat p_i$, the \emph{internal segments} of $\hat p_i$, the \emph{last segment} of $\hat p_i$, and finally the \emph{last tile} of $\hat p_i$.
                Failure in one of these parts causes a \textbf{preemption} of $r_i$.
            \item Packet request $r_i$ \textbf{arrives} to its destination $b_i$ if it is not rejected or preempted.
        \end{enumerate}
\caption{The deterministic framework. The algorithm receives a sequence of packet requests over the network $G=(V,E)$ and it either rejects, injects, or preempts these packet requests. A packet arrives to its destination if it is not rejected or preempted. The deterministic algorithm executes the \IPP\ algorithm as a sub-procedure.}\label{alg:algDet}
\end{algorithm}

\section{The One Dimensional Case}\label{sec:alg}
In this section we present the details of Algorithm~\ref{alg:algDet} for $d=1$.
We refer to Algorithm~\ref{alg:algDet} by \alg.

\paragraph{Parameters.}
The parameters of the uni-directional line network $G$ are:
$n$ nodes, buffer size $B$ in each node, and the capacity
of each link is $c$. We assume that $B,c\in [3,\log n]$.
Let $\pmax=2n  \cdot \left(1+n \cdot
\left(\frac{B}{c}+1\right)\right) = O(n^2\cdot \log n)$.
Let $k\eqdf \lceil \log (1+3\pmax) \rceil$. The length of a
tile's side is $k$.

\begin{proposition}\label{prop:tiling}
If $B,c\leq \log n$, then
\begin{inparaenum}[(i)]
\item$k=O(\log n)$, and
\item the capacity of each edge in the sketch graph is at most
  $k\cdot \max\{B,c\} = O(\log^2 n)$.
\end{inparaenum}
\end{proposition}

\subsection{Reduction to Online Integral Path Packing}\label{sec:reduce}

\paragraph{Downscaling of Capacities.}
We regulate the number of paths that traverse each edge and node in the sketch graph
by downscaling capacities.  There are three types of capacities: (1)~edges between
tiles are assigned unit capacities, (2)~incoming edges to sink nodes are unchanged
and remain with infinite capacities, and (3)~each tile is assigned two units of
capacity\footnote{In the case of $d$-dimensional grid, the capacity of a tile is
  $d+1$. This saves a factor of $d$ in the competitive ratio.}.

To apply a reduction to integral path packing, we reduce node capacities to edge
capacities. Namely, each node $s \in V(S)$ is split to two ``halves'' $s_{in}$ and
$s_{out}$. After the split, edges are ``redirected'' as follows: the incoming edges
of $s$ enter $s_{in}$ and the outgoing edges of $s$ emanate from $s_{out}$. We add an
additional edge called an \emph{interior edge} between $s_{in}$ and $s_{out}$.  All
interior edges are assigned two units of capacity (see Figure~\ref{fig:dsketchcaps}).  We refer to the augmented sketch
graph with these capacities as the \emph{$\{1,2,\infty\}$-sketch graph}. We denote
the $\{1,2,\infty\}$-sketch graph by $\hat S$.
Let $\hat c:E(\hat S) \rightarrow \{1,2,\infty\}$ denote the downscaled capacity function of the $\{1,2,\infty\}$-sketch graph $\hat S$.

Note that, since nodes are split and sinks are added, we need to increase the maximum path length to $\pmax\leftarrow 2\cdot \pmax+1$.

\begin{figure}[h]
  \centering
\includegraphics[width=0.45\textwidth]{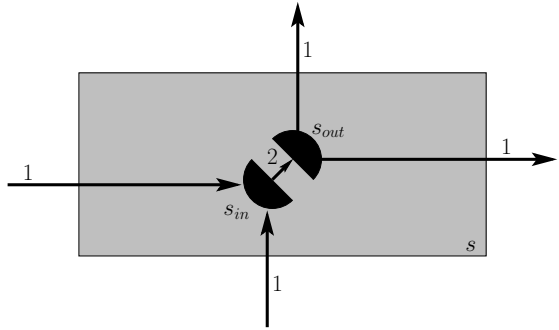}
\caption{
Capacity assignment in the $\{1,2,\infty\}$-sketch graph $\hat S$. Unit capacities are assigned to sketch edges and capacity of $2$ is assigned to interior edges.}
\label{fig:dsketchcaps}
\end{figure}

\paragraph{The Reduction.}
A request $r_i=(a_i,b_i,t_i)$ to deliver a packet is reduced to a path request $\hat r_i$ in the $\{1,2,\infty\}$-sketch graph $\hat S$. The source of the path request $\hat r_i$ is the vertex  $s_{in}$, where the vertex $(a_i,t_i)$ is in tile $s$. The destination of the path request is simply the sink node $\hat b_i$.

The sole purpose of the sink node is for a clean reduction to path packing. Once the \IPP\ algorithm returns the sketch path $\hat p_i$, the sink node is removed from $\hat p_i$, and the last tile in the sketch path is regarded as the end of the sketch path.

Theorem~\ref{thm:IPP} implies that the \IPP\ algorithm returns an integral packing of paths in $\hat S$ that is
$(2,k)$-competitive with respect to the optimal fractional path packing in $\hat S$. The length of each path in the packing is at most $\pmax$.

\subsection{Detailed Routing}\label{sec:detailed}
This section deals with the translation of paths in the sketch graph to paths in the space-time graph. This translation, called \emph{detailed routing}, is adaptive and computed in a distributed on-the-fly fashion.
The detailed path respects the sketch path in the sense that it traverses the same tiles and bends only where the sketch path bends.
Note that, some of the packets are dropped during detailed routing.

More formally, the goal in detailed routing is to compute a (detailed) path $p_i$ in
the space-time graph $G^{st}$ given a sketch path $\hat p_i$ in the $\{1,2,\infty\}$-sketch graph $\hat S$.
The projection of $p_i$ on
$\hat S$ equals $\hat p_i$.

\subsubsection{Preliminaries}

\paragraph{Terminology.}
A \emph{bend} in the sketch path is a node in which the sketch path changes direction, i.e., vertical to horizontal or horizontal to vertical.

A \emph{segment} of a path in a grid is a maximal subpath, all the vertices of which belong to the same row or column of the grid.
A segment is \emph{special} if it is the first or the last segment of a path. Otherwise, it is an \emph{internal} segment.

We refer to the side through which the detailed path enters a tile as the \emph{entry side}.
Similarly, we refer to the side through which the detailed path exits a tile as the \emph{exit side}.

\paragraph{Packing Intervals Online.}\label{sec:intervalp}
The problem of packing intervals in a line is defined as follows.
  \begin{enumerate}
  \item Input: A set $I=\{p_i\}_{i=1}^r$, where each $p_i$ is an open interval
    $(a_i,b_i) \subseteq (1,n)$.\footnote{We consider open intervals rather than
      closed intervals. One could define the problem with respect to closed
      intervals, but then instead of requiring disjoint intervals in the packing,
      one would need to require that intervals may only share endpoints.}
          \item Output: A maximum cardinality subset $I' \subseteq I$ of pairwise
            disjoint intervals.
    \end{enumerate}
    In the online setting, we assume that the
    intervals appear one by one, and that $a_1\leq a_2\leq \cdots \leq a_r$. The online algorithm must maintain a maximum subset $I'$ such that
    (i)~$I'$ is a subset of the prefix of the intervals input so far, and (ii)~the
    intervals in $I'$ are pairwise disjoint.

    The online algorithm is based on an optimal algorithm for maximum independent
    sets in interval graphs~\cite{GLL}.  Upon arrival of an interval $p_i=(a_i,b_i)$,
    the algorithm proceeds as follows: (1)~If $p_i$ does not intersect the intervals
    in $I'$, then $p_i$ is added to $I'$.  (2)~Else, $p_i$ intersects an interval
    $p_j=(a_j,b_j)$. If $b_i>b_j$, then $p_i$ is rejected (namely, $I'$ remains
    unchanged). Otherwise, if $b_i\leq b_j$, then $p_i$ preempts $p_j$ (namely,
    $I'=(I'\cup\{p_i\}) \setminus \{p_j\}$).

    Note, that this online algorithm can be executed in a distributed fashion in a
    line. Namely, the local input of each processor $a_i$ is the interval $p_i=(a_i,b_i)$ (or the empty input). Additionally, $a_i$ receives  $I'$ from its neighbor $a_{i-1}$. Now, $a_i$ can verify by itself whether to preempt an interval from $I'$ and accept $p_i$ or to reject $p_i$. After $a_i$ completes his local computation, $a_i$ sends $I'$ to its neighbor $a_{i+1}$.


\paragraph{Partitioning of Detailed Routing.}
Detailed routing is partitioned into at most three parts%
\footnote{Degenerate cases of detailed routing consist of two parts or just a single part; for example, detailed routing of requests whose sketch path is a single tile consists only of part (III).
}
, as follows (See Figure~\ref{fig:detailedpath}).
\begin{enumerate}[(I)]
\item Special segments,
\item Internal segments, and
\item Last tile: detailed routing in the last tile deals with routing the request from the point that it enters the last tile till a copy of the  destination vertex within the tile.
\end{enumerate}
Preemptions may occur in parts (I) and (III) of the detailed routing. Preemptions are caused by conflicts between detailed routing of packets that belong to the same part. Namely, a special segment can only preempt another special segment. Similarly, detailed routing in the last tile preempts only routes that end in the same tile.

\paragraph{Reservation of Capacities.}
The algorithm reserves one unit of capacity in each edge $e\in
E^{st}$ for each part of detailed routing. This is the reason for the requirement that $B,c \geq 3$.
Note that the algorithm is wasteful in the sense that it only uses $3$ units of capacity in each edge.
We refer to each of these $3$ units of capacity as a \emph{track}, i.e., each part uses a different track.

\subsubsection{Detailed Routing in Special Segments}\label{sec:first detailed}
Consider the first segment of a sketch path $\hat{p}_i$ (see
Fig.~\ref{fig:sketchpath}).  The detailed routing corresponding to this segment is a
straight path that starts in the source-vertex $(a_i,t_i)$ and ends in the tile in
which $\hat{p}_i$ bends for the first time. As there may be contention for capacity
allocated for special segments, detailed routing needs to decide which request is
dropped. We reduce the problem of routing the first segment of detailed paths to the
problem of packing intervals in a line (described in detail in
Section~\ref{sec:intervalp}).

A separate reduction to interval packing in a line takes place for every row and column of the
untilted space-time grid.

Detailed routing in the last segment of $\hat{p}_i$ (before the last tile) is similar.
Consider a last segment of a sketch path $\hat{p}_i$ that starts in tile $s_1$ and
ends in tile $s_2$.  The detailed routing of a last segment
begins in the entry side of $s_1$ that is reached by the detailed routing of
the previous segment, and ends in the entry side of $s_2$. Between these two
endpoint, detailed routing is along a straight path. As in the case of detailed
routing of the first segment, routing in the last segment is reduced to interval
packing in a line.

Consider a sketch path $\hat p_i$ whose first bend is in tile $s$.  If the detailed
routing of the first segment of $p_i$ is not preempted before it enters the tile $s$,
then $r_i$ is not preempted before the first bend. Indeed, there are two types of
conflicting requests whose first segment conflicts with the first segment of $r_i$
depending on the location of the source vertex (either before or after the entry to
tile $s$). If the source vertex of $r_j$ appears before the entry to $s$, then $r_i$
``wins'' and $r_j$ is preempted. If the source vertex of $r_j$ appears after the entry
to $s$, then $r_i$ ``wins'' again because $r_j$ requests an interval that ends
outside the tile $s$ while $r_i$ requests an interval that ends in tile $s$. We also
need to consider a conflict with a last segment of a request $r_j$: (1)~If $r_j$ ends
inside $s$, then it must also begin in $s$ (because it is not possible for $r_i$ and
$r_j$ to enter the tile through the same edge). If $r_j$ begins and ends $s$, then it
is routed using only the third track (reserved for detailed routing in the last tile)
and $r_j$ does not conflict with the first segment of $r_i$. (2)~If $r_j$ ends outside
$s$, then it is preempted by $r_i$ because $r_i$ requests an interval that ends
inside $s$.

\begin{figure}
  \centering
  \subfloat[The sketch path $\hat{p}_i$]{
    \includegraphics[width=0.47\textwidth]{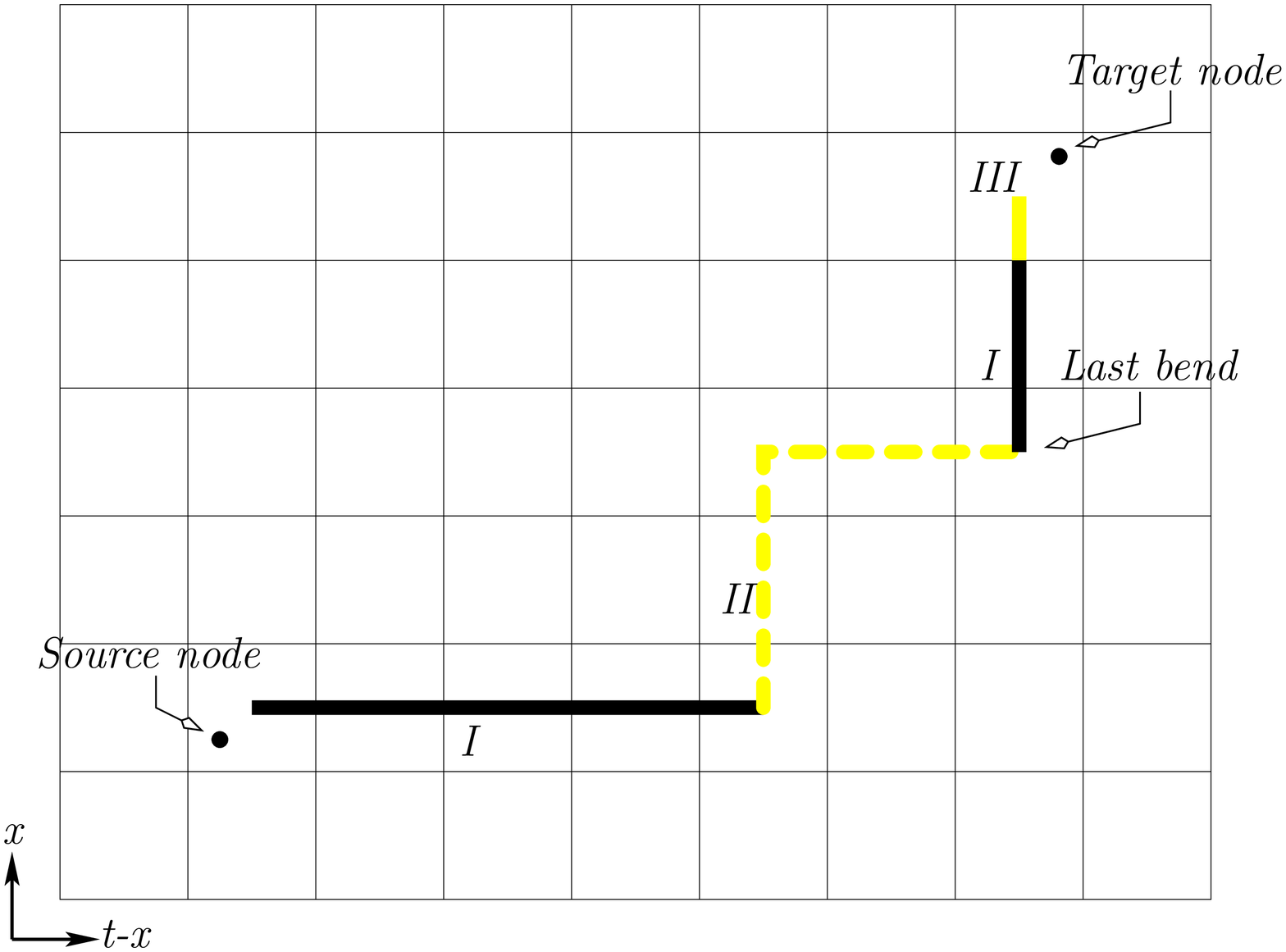}
  \label{fig:sketchpath}}
  \qquad
  \subfloat[The detailed path $p_i$]{
    \includegraphics[width=0.44\textwidth]{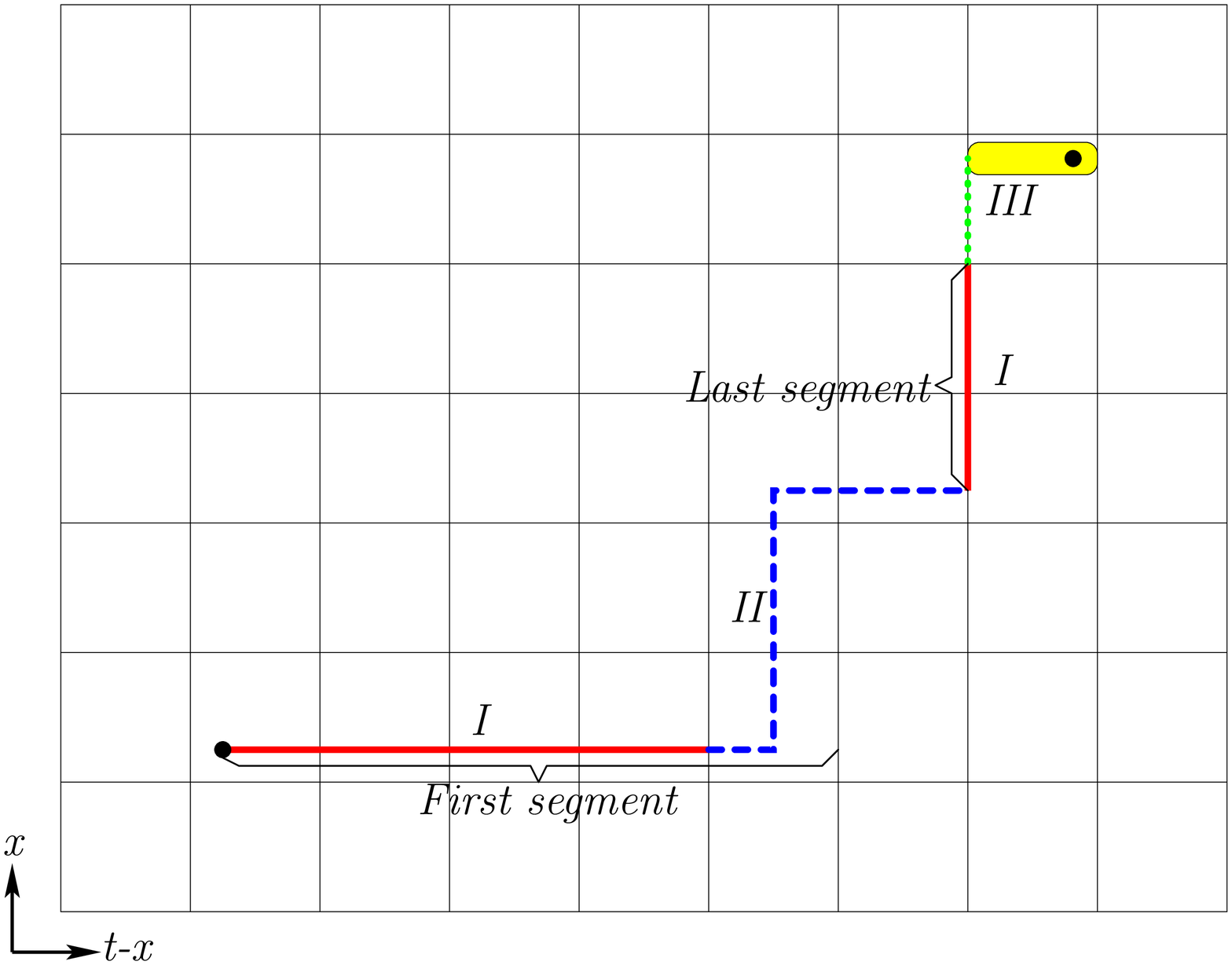}
  \label{fig:detailedpath}}
  \caption{The untilted space-time graph $G^{st}$ is partitioned into tiles depicted by square rectangles. These tiles are the vertices of the $\{1,2,\infty\}$-sketch graph $\hat S$, in fact, two neighboring squares correspond to two neighboring vertices in $\hat S$. (a) The sketch path $\hat{p}_i$ is overlayed on $G^{st}$. We partition $\hat{p}_i$ into three parts: (I)~first and last segments, which are depicted by solid segments, (II)~internal segments, which are depicted by dashed segments, and (III) routing in the last tile, which is depicted by a grey line. The source node of the packet request is in the first tile of the sketch path, the target node of the packet request is in the last tile of the sketch path. (b) The detailed path $p_i$ is depicted by a thin line that traverses the same tiles traversed by the sketch path $\hat{p}_i$. The detailed routing  of the first segment is depicted by the horizontal line emanating from the source  node. The dashed line depicts the detailed routing after the first segment. The detailed routing of the last segment takes a turn on the entry side of the tile that contains the last bend. The detailed routing in the last tile is depicted by an straight dotted thin line.  The space-time copies of $b_i$ are depicted by the grey rectangle that surrounds the target node.
  The intervals that are input to the interval packing algorithm are depicted by braces.}
  \label{fig:ddetail}
\end{figure}

\subsubsection{Detailed Routing in Internal Segments}\label{sec:detailed internal}
Detailed routing of internal segments takes place in a tile as follows. Fix a node
$v$. The node $v$ has two incoming edges and two outgoing edges. We denote these
edges by $horz_{in}, vert_{in}$ and $horz_{out},vert_{out}$.  We refer to the
request that traverses an edge $e$ by $e.r$.  For example, $horz_{in}.r$ is the name
of the request that enters $v$ via the horizontal edge. If an edge $e$ is not
assigned to a request, then we set $e.r$ to null. The rules for detailed routing of
these paths are as follows:
\begin{enumerate}
\item If one of the incoming edges $e$ is not assigned to a request, then the other
  edge $e'$ (if $e'.r$ is not null) chooses the outgoing edge according to its
  exit side.
\item (Precedence to straight traffic.) Else, if the exit side of $horz_{in}.r$ is east or the
  exit side of $vert_{in}.r$ is north, then the paths continue without a bend,
  namely, $horz_{out}.r\gets horz_{in}.r$ and $vert_{out}.r\gets vert_{in}.r$.
\item (Simultaneous bends.) Else, a knock-knee bend takes place, namely,
  $horz_{out}.r\gets vert_{in}.r$ and $vert_{out}.r\gets horz_{in}.r$. (see
  Figure~\ref{fig:knockknee}).
\end{enumerate}

We claim that detailed routing in an internal segment always succeeds.  If the
detailed path is headed towards its exit side (e.g., traverses the tile without a
bend), then detailed routing gives it priority so that it reaches its exit
side.  If the sketch path bends in the tile, then the detailed path must encounter
either a null path or another detailed path that also bends in the tile (in which
case the path takes the required turn). This is true because, otherwise, there would
be more than $k$ paths that exit the tile from the same side, contradicting the
congestion guarantee by the \IPP\ algorithm (that at most $k$ paths traverses the edges
between tiles).

\begin{figure}[H]
  \centering
    \includegraphics[width=0.4\textwidth]{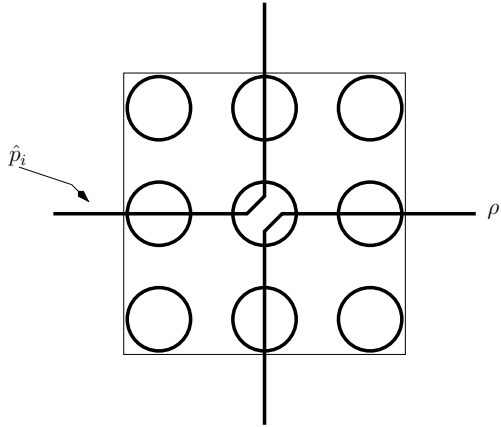}
  \caption{A knock-knee bend in detailed routing in $G^{st}$. Space-time nodes are depicted by white circles. The detailed route of $\hat p_i$ makes a turn in the vertical direction, thus freeing the suffix of the row $\rho$. The conflicting detailed route takes a turn in horizontal direction, thus freeing the suffix of the column in the vertical direction.}
    \label{fig:knockknee}
\end{figure}

We now deal with transitions from part (I) to part (II) of detailed routing.
Recall, that each part of the detailed path uses a different track.
Consider a sketch path $\hat p_i$ whose first bend is in tile $s$.
If the detailed routing of $\hat p_i$ reaches $s$, then it is not preempted by
another special segment (see Sec.~\ref{sec:first detailed}).
As in detailed routing in internal segments, the detailed route of $\hat p_i$ in tile
$s$ bends when it meets a null path or a detailed path that also wants to bend.
The same argument shows that such a bend is always successful.
After the bend, the path transitions from the first track to the second track.

We conclude that detailed routing is always successful in internal
segments.

\subsubsection{Detailed Routing in the Last Tile}
We refer to requests whose sketch path is a single tile as \emph{near} requests.
Note that detailed routing of a near request consists only of part (III).

Detailed routing in the last tile routes a path along a straight vertical path from
the entry point to the row in the tile that corresponds to the destination node.
Note that if the destination vertex of $r_i$ is $b_i$, then it suffices to route the
path to one of the space-time copies of $b_i$. Hence, every copy of $b_i$ in the tile
is a valid destination. Contentions occur only in each column, and a path with a
closest destination preempts the conflicting paths.

\subsection{Analysis of the Algorithm for $d=1$}\label{sec:analysis}

Recall that the length of a tile's side is $k = \lceil \log
(1+3\pmax) \rceil$. Moreover, in the case where  $B,c \in
[3,\log n]$, it follows that $k = O ( \log n)$.

\begin{theorem}\label{thm:alg}
  The competitive ratio of the algorithm for uni-directional line networks is $O(\log^5 n)$ provided
  that $B,c \in [3,\log n]$.
\end{theorem}

\begin{proof sketch}{Theorem~\ref{thm:alg}}
The algorithm starts with the path packing algorithm \IPP\ over the $\{1,2,\infty\}$-sketch graph.
This means that capacities are reduced by a factor of at most $O(k^2 \cdot \max\{B,c\})=O(k^3)$ (by the capacity assignment ``inside'' a tile and ``between'' tiles).
The fact that path lengths are bounded by $\pmax$ reduces the throughput only by a constant factor. The throughput of algorithm \IPP\ is $O(1)$-competitive.

Detailed routing succeeds in routing at least a $k^2$ fraction of the sketch paths. There are two causes for loss of packets: routing of special segments and routing in the last tile.
Routing of special segments (i.e., first and last segment) succeeds for a fraction of $1/k$.
we show that the success rate is not multiplied and that the success rate for special segments is $1/2k$. Routing in the last tile succeeds for a fraction of $1/2k$ per tile.
Putting things together we get a competitive ratio of $O(k^5)$, as required.
\end{proof sketch}

Note that the Theorem~\ref{thm:alg} actually applies for $B,c \in [3,O(\log
n)]$. The constant in the $O(\log n)$ linearly affects the constant in the competitive ratio of the algorithm.

\paragraph{Notation.}
Let $R$ be a fixed sequence of packet requests introduced by the adversary.
Let $R_s \subseteq R$ denote the set of requests whose sketch path ends in tile $s$.
For every $X \subseteq R$ and for every tile $s$ let $X_s \eqdf X \cap R_s$.
We interpret requests in $R$ as path requests in $G^{st}$.  Let $\opt$ (respectively
$\opt_f$) denote a maximum integral (respectively fractional) packing of paths from
$R$ in $G^{st}$.
Let $\IPP(R)$ denote the set of requests that algorithm \IPP\ injected when
given input $R$.  For brevity, we denote $\IPP(R)$ simply by $\IPP$.
Similarly, let $\alg$ denote the set of requests that \alg\ routed to their
destination.  Let $\IPP'\subseteq \IPP$ denote the set of requests that are not
preempted before they reach the entry side of their last tile.
(Note that $\alg \subseteq \IPP' \subseteq \IPP \subseteq R$.)
Let $f^*$ denote an optimal fractional flow with respect to $R$ over the sketch graph $S$.
Let $f^*_{\{1,2,\infty\}}$ denote an optimal fractional flow with respect to $R$
over the $\{1,2,\infty\}$-sketch graph $\hat S$.
(Note that $\opt$ and $\opt_f$ are packings of paths in $G^{st}$, while $f^*$ and
$f^*_{\{1,2,\infty\}}$ are packings in sketch graphs.)
Let $\opt_f(R \mid \pmax)$ denote an  optimal fractional path packing in $G^{st}$ with respect to $R$ under the constraint that each request is routed along a path of length at most $\pmax$.
Let $f^*(R \mid \pmax)$ denote an
optimal fractional flow in the sketch graph $S$ with respect to
$R$ under the constraint that flow paths have a length of at most $\pmax$.
Let $f^*_{\{1,2,\infty\}}(R \mid \pmax)$ denote an
optimal fractional flow in the $\{1,2,\infty\}$-sketch graph $\hat S$ with respect to
$R$ under the constraint that flow paths have a length of at most $\pmax$.
Let $|g|$ denote the throughput of flow $g$.

\medskip
\noindent
We now present a detailed proof of Theorem~\ref{thm:alg}, based on the following propositions.
\begin{proposition}\label{prop:opt}
$|f^*(R\mid\pmax)| \geq |\opt_f(R\mid\pmax)|$.
\end{proposition}
\begin{proof}
  Consider a fractional packing $h$ of paths in $G^{st}$ in which
  paths lengths are bounded by $\pmax$.  Let $g$ denote the flow in
  sketch graph $S$ where $g(e)$ is simply the sum of the flows of
  $h$ along the edges in $G^{st}$ that are coalesced to $e$ in
  $S$. Clearly, $|g|=|h|$. We claim that $g$ is a feasible
  fractional flow in the sketch graph $S$ whose flow paths are not
  longer than the flow paths in $h$. (In fact, they are shorter by a factor of $k$.)

  We show that the flow $g$ satisfies the capacity constraints in $S$
  as follows. If $e$ is a sketch edge between tiles, then, by
  linearity, the capacity constraint is satisfied.  We now focus on
  interior edges.  The amount of flow in $h$ that traverses a tile in
  $G^{st}$ is bounded by the sum of the capacities of the
  edges in the tile, namely, it is at most $(B+c)\cdot k^2$. It
  follows that the amount of flow in $g$ that traverses a node (that
  corresponds to a tile) in the sketch graph is bounded by the node's
  capacity (which equals $2\cdot k^2 \cdot (B+c)$). We conclude that $g$ is a feasible flow in $S$, and the proposition follows.
\end{proof}

\begin{proposition}\label{prop:scaled}
    $k^{2}\cdot(B+c)\cdot |f^*_{\{1,2,\infty\}}(R\mid\pmax)| \geq |f^*(R\mid\pmax)| \geq |f^*_{\{1,2,\infty\}}(R\mid\pmax)|$
\end{proposition}

\begin{proof}
  Recall that $f^*$ is a maximum flow in the sketch graph $S$ while
  $f^*_{\{1,2,\infty\}}$ is a maximum flow in $\hat{S}$.
The proof is a direct consequence of the following bounds between capacities in $S$
and in $\hat S$.

  For every edge $e$ that is both in $S$ and in $\hat S$, we have
  \begin{equation}\label{eq:cap S}
    k\cdot (B+c) \cdot \hat c(e) \geq c(e) \geq \hat c(e).
  \end{equation}

 For every node $s$ that corresponds to a tile, we have
  \begin{equation}\label{eq:cap S}
  c(e)= k^2 \cdot (B+c)\cdot \hat c(e).
  \end{equation}
\end{proof}

\begin{proposition}\label{prop:fipp}
    $|\IPP| \geq \left(\frac{1}{2\cdot k^{2}\cdot(B+c)}\right) \cdot |f^*(R\mid \pmax)|$
\end{proposition}
\begin{proof}
By    Theorem~\ref{thm:IPP} (i.e., $(2,k)$-competitiveness of \IPP),
\begin{align*}
  |\IPP| &\geq \frac {1}{2} \cdot f^*_{\{1,2,\infty\}} (R \mid \pmax ).
\end{align*}
Downscaling of capacities implies
\begin{align*}
f^*_{\{1,2,\infty\}}(R \mid \pmax)
  &\geq \left(\frac{1}{k^{2}\cdot(B+c)}\right) \cdot |f^*(R\mid \pmax)|,
\end{align*}
and the proposition follows.
\end{proof}

The following proposition proves that a fraction of at most $(1-\frac{1}{2k})$ of the
requests in \IPP\ are preempted before they reach their last tile.
\begin{proposition}\label{prop:preemptions}
    $|\IPP'| \geq \frac{1}{2k} \cdot |\IPP|$
\end{proposition}

\begin{proof}
  Consider a row or a column $L$ of nodes in $G^{st}$. Let $R\cap L$ denote the set of requests that contain special
  segments that compete over edges in $L$. From the point of view of
  $L$, each request $r_i\in R\cap L$ is a request for an interval
  $I_i\subseteq L$.
  As described in Section~\ref{sec:intervalp}, the detailed routing of the requests $R\cap L$ along  $L$ simulates an optimal interval packing algorithm. In
  particular, the simulation has the property that if an interval
  $I_i=(a_i,b_i)$ preempts an interval $I_j=(a_j,b_j)$, then the intervals overlap
  and $b_i \leq
  b_j$. Hence, the edge $(b_i-1,b_i)$ is in $I_j$.

  Focus on preemptions that occur during the detailed routing of first segments (the
  case of last segments is similar). Consider the ``forest of preemptions'' over the
  intervals, where the set of intervals that were preempted by $I_i$ are children of
  $I_i$.  We claim that if interval $I_j$ is a descendant of $I_i$ in this forest,
  then the edge $(b_i-1,b_i)$ is in $I_j$.  The proof is by induction on the distance
  between $I_i$ and $I_j$ in the forest of preemptions. The induction basis holds for
  a child $I_j$ by the discussion above.  Suppose that $I_k$ preempted $I_j$ (hence
  $b_k\leq b_j$). Since $I_k$ is a descendent of $I_i$, by the induction hypothesis
  $(b_i-1,b_i)$ is an edge in $I_k$. Because $I_j$ is preempted by $I_k$ in a vertex
  to the left of $b_i$, it follows that the edge $(b_{i}-1,b_i)$ is in $I_j$, as
  required.  By Theorem~\ref{thm:IPP}, the load induced by $\IPP$ on each
  $\{1,2,\infty\}$-sketch edge is at most $k$.  Therefore, the maximum number of
  proper descendants of $I_i$ in the forest is $(k-1)$ (not including $I_i$).

  Consider a bipartite graph of preemptions over $\IPP'\cup (\IPP\setminus \IPP')$
  (now we consider both first segments and last segments). There is an edge
  $(r_i,r_j)$ if the request $r_i\in \IPP'$ is an ancestor of the request $r_j\in
  (\IPP\setminus \IPP')$ in the forest of preemptions corresponding to detailed
  routing. Since a preempted request is preempted only once, the degree of the nodes
  in $\IPP\setminus \IPP'$ is one.  Recall that each sketch path contains at most $2$
  special segments.  By the discussion above, the degree of a node in $\IPP'$ is
  bounded by $2\cdot(k-1)$. 
  By counting edges in the bipartite graph, we conclude that $|\IPP'| \cdot 2
  \cdot(k-1) \geq |\IPP\setminus \IPP'|$, and the proposition follows.
\end{proof}

The following proposition states that a fraction of at least $1/(2k)$ of the requests
that reach their last tile are successfully routed.
\begin{proposition}\label{prop:last}\label{prop:Rs}
    $|\alg| \geq \frac{1}{2k} \cdot |\IPP'|$
\end{proposition}

\begin{proof}
  Since $\{\IPP'_s\}_{s \in V(S)}$ is a partition of $\IPP'$ and  $\{\alg_s\}_{s \in V(S)}$ is a partition of $\alg$, it suffices to prove
  that $|\alg_s| \geq \frac{1}{2k}\cdot |\IPP'_s|$ for every tile $s$.

  Fix a tile $s$.  Every sketch path of a request in $\IPP'_s$ traverses the interior
  edge of $s$ in $\hat S$ whose capacity is $2$. Theorem~\ref{thm:IPP} implies that
  this capacity is violated by at most a factor of $k$, hence $|\IPP'_s| \leq 2k$.

  Detailed routing in the last tile successful routes at least one request from
  $\IPP'_s$ if $\IPP'_s \neq \emptyset$, and the proposition follows.
\end{proof}

\noindent
We now put things together to complete the proof of Theorem~\ref{thm:alg}.
\begin{proof}[proof of Theorem~\ref{thm:alg}]
The proof is as follows.
  \begin{align*}
    |\alg|  &\geq  \frac{1}{2k} \cdot |\IPP'|&\text{(by Prop. ~\ref{prop:last})}\\
     &\geq  \frac{1}{2k} \cdot \frac{1}{2k} \cdot |\IPP|&\text{(by Prop. ~\ref{prop:preemptions})}\\
    &\geq \left(\frac{1}{8\cdot k^{4}\cdot (B+c)}\right) \cdot  |f^*(R\mid \pmax)|&\text{(by Prop. ~\ref{prop:fipp})}\\
    &\geq \left(\frac{1}{8\cdot k^{4}\cdot (B+c)}\right) \cdot  |\opt_f(R\mid \pmax)| &\text{(by Prop. ~\ref{prop:opt})}\\
    &\geq \left(\frac{1}{8\cdot k^{4}\cdot (B+c)}\right)  \cdot \frac 12 \cdot \left(1-\frac 1e \right) \cdot |\opt_f(R)|&\text{(by Lemma. ~\ref{lemma:nB})}\\
    &\geq  \Omega\left(\frac{1}{k^{4}\cdot(B+c)}\right) \cdot
    |\opt|\:.
  \end{align*}
The last line holds because every integral path packing is also a fractional one.
The theorem follows.
\end{proof}

\subsection{Requests With Deadlines}\label{sec:d_i}
In this section we present the modification needed to deal with packet
requests with deadlines.  The change to the algorithm is in the
reduction to online integral path packing (see Section~\ref{sec:reduce}), i.e., we need to change the sink node in the reduction as described below.

\newcommand{\sink}{\textit{sink}}

\paragraph{Adding Sink Nodes for Requests with Deadlines.}
A request to deliver a packet is of the form $r_i=(a_i,b_i,t_i,d_i)$,
where $d_i$ is the deadline.  In terms of a path request in the space-time
graph $G^{st}$, this means that we need to assign a path from
$(a_i,t_i)$ to a vertex $(b_i,t')$, where $t_i \leq t'\leq d_i$.
Thus, the destination is a set of vertices rather than one specific
vertex.
We connect this set of destinations to a new sink. Formally, for every
request $r_i$, introduce a new vertex $\sink_i$ and connect every vertex
in $\{(b_i,t')\}_{t'=t_i}^{d_i}$ to $\sink_i$ with an edge of infinite
capacity.

Now, a packet request $r_i=(a_i,b_i,t_i,d_i)$ is reduced to a path
request in the $\{1,2,\infty\}$-sketch graph from the half-tile
$s_{in}$ (where the tile $s$ contains $(a_i,t_i)$) to $\sink_i$.
A path from $(a_i,t_i)$ to $\sink_i$ contains at most $d_i-t_i+1$ edges. We still bound the path length by $\pmax$, as before, to obtain a load of $O(\log \pmax)$ by $\IPP$.

We claim that a request that is not preempted by detailed routing reaches its destination on time.
To see this fix a packet request $r_i$ that is not preempted by detailed routing, and
let $\hat p_i$ denote its sketch path. Let $s$ denote the tile in which $\hat
p_i$ ends.  We now show that the detailed path $p_i$ ends in a vertex $(b_i,t)$ such
that $t \leq d_i$.  There are $3$ cases (see
Figure~\ref{fig:dltile}): (1)~$p_i$ enters $s$ via a last segment from
the south-west corner of $s$, (2)~$p_i$ enters $s$ via a first segment from the west,
or (3)~$p_i$ enters $s$ via a first segment from the south\footnote{ Note that
  cases (2) and (3) are degenerate cases in the sense that the detailed routing
  consists only of the a first segment and routing in the last tile.}.  In the first two cases, $p_i$ enters $s$ and moves north
until it reaches a copy of $b_i$.  The copy $(b_i,t')$ of $b_i$ that is reached must
satisfy $t' \leq d_i$ if $(b_i,d_i)$ is in the tile.  Indeed, because $s$ is the last
tile of $\hat p_i$, the copy of $b_i$ in the leftmost column of $s$ lies below the
``time-zone'' $\{(x,d_i-x)\}_x$ in the untilted space-time graph. Moreover, the entry point
of $p_i$ to tile $s$ lies below this copy of $b_i$ (if it were above this copy of
$b_i$, then it has already reached $b_i$).  In the third case, $p_i$ enters via the
south side. This means that (before entering $s$) $p_i$ consists only of a first
segment, i.e., starting from its arrival the packet was forwarded and was not
buffered at all.  Since the deadlines are ``feasible'', i.e., the deadline $d_i \geq
t_i + dist(a_i,b_i)$, where $dist(a_i,b_i)$ is the distance between $a_i$ to $b_i$.
The packet keeps moving north and reaches the copy of $b_i$ at time $t_i +
dist(a_i,b_i)$. It follows that the packet reaches its destination on time in this
case as well.  We conclude that requests that are not preempted reach their
destination on time, as required.

\begin{figure}[H]
  \centering
    \includegraphics[width=0.5\textwidth]{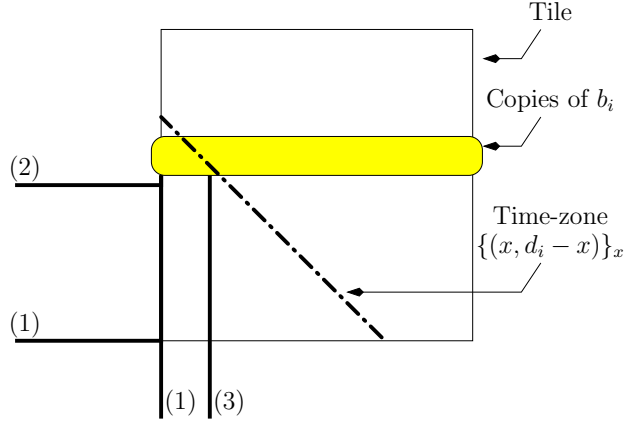}
  \caption{The $3$ possible starting points of detailed routing in a tile.}
    \label{fig:dltile}
\end{figure}

\section{Generalizations}\label{sec:generalizations}
In this section we present a generalization of the algorithm to the $d$-dimensional case as well as extensions to the special cases: bufferless grids and grids with large buffers\textbackslash capacities.

\paragraph{The $d$-Dimensional Case.}\label{sec:algd}
The following modifications are needed to extend the algorithm to $d$-dimensional grids.
\begin{enumerate}[(1)]
\item $k = \lceil \log (1+3\pmax) \rceil$, where in the
    $d$-dimensional case
  $$\pmax\triangleq 2 \cdot \diam(G)\cdot \left(1+n\cdot \left(\frac{B}{c} + d \right)\right) \:.$$
    In the case where $B,c \in [3,\log n]$, it follows that $k = O(\log
    n)$.

\item Apply tiling with side length $k$, e.g., a face of a cube contains $k^d$ vertices.

\item Similarly to the $1$-dimensional case, the sketch graph also has node capacities for nodes that correspond to tiles (i.e.,
not sinks). The capacity of every node that corresponds to a tile is $c(s)=(d+1)\cdot k^{d+1}\cdot (B+d\cdot c)$.
Edges in the sketch path have unit capacities.

\item Similarly to the definition of $\{1,2,\infty\}$-sketch graph, we define the
  $\{1,d+1,\infty\}$-sketch graph by assigning a capacity of $d+1$ (instead of $2$)
  to the interior edges.

\item Detailed routing of internal segments is generalized as follows. Each node has
  $d+1$ incoming edges and $d+1$ outgoing edges. Fix a node $v$. Let
  $in_1,\ldots, in_{d+1}$ denote edges that enter $v$.  Similarly, let
  $out_1,\ldots,out_{d+1}$ denote edges that exit $v$.  Detailed routing in $v$
  proceeds as follows:
For every $j\in[1,d+1]$, let $\ell_j$ denote the exit side of request
    $in_j.r$ in the tile $s$ that contains $v$.
  \begin{enumerate}
  \item (Precedence to straight paths.) If $\ell_j=j$, then $out_j.r=in_j.r$.
  \item (Try next crossing.) Else, if the exit side of $in_{\ell_j}.r$ is not $j$ or null, then
    $out_j.r=in_j.r$.
  \item Else, if $in_{\ell_j}.r=j$ or ($in_{\ell_j}.r=null$ and $j$ is the smallest index $j'$ for which $in_{j'}.r=\ell_j$), then a knock-knee takes place: $out_{\ell_j}.r=in_j.r$ and $out_{j}.r=in_{\ell_j}.r$.
  \item (Try next crossing.) Else, $out_j.r=in_j.r$.
  \end{enumerate}

  The key observation for detailed routing in an internal segment is that if a
  request $r_i$ fails to bend at node $v$, then another request proceeds in $v$
  toward its exit side (in the tile that contains $v$). Thus, as a request $r_i$
  continues to try to turn in the next crossing, it crosses a new request that will
  exit the tile successfully. Since the number of requests in $\IPP$ that traverse the same sketch edge is at most $k$, it follows that $r_i$ is bound to find a crossing in which it turns
  toward its exit side.

\end{enumerate}
The following theorem bounds the competitive ratio of the algorithm for general dimensionality $d$. The proof of Theorem~\ref{thm:algd} is outlined in Appendix~\ref{sec:proofdd}.

\begin{theorem}\label{thm:algd}
  The competitive ratio of the algorithm for $d$-dimensional grid
networks is $$O\left(k^{d+3} \cdot(B+d\cdot c) \right)=O\left(
\log^{d+4} n \right)$$ provided that $B,c \in [3,\log n]$.
\end{theorem}
\label{sec:extend}
\paragraph{Bufferless Grids.}
For the case $B=0$ and $c\geq 3$ (no upper bound on $c$), we obtain
the following result. The proof of the following theorem is sketched in Appendix~\ref{sec:proofs}.
\begin{theorem}\label{thm:bufferless}
  There exists an online deterministic preemptive algorithm for packet
  routing in bufferless $d$-dimensional grids with a competitive ratio of $O(\log ^{d+2}
  n)$.
\end{theorem}

In the one dimensional case without buffers, the optimality of online interval
packing implies that the nearest-to-go policy~\cite{AKOR} is optimal.
\begin{proposition}
Nearest-to-go is an optimal policy for packet routing in a line when $B=0$.
\end{proposition}

\paragraph{Large Buffers \& Large Link Capacities.}\label{sec:largeBc}
In this section we consider the case that the size of the buffers and the capacities
of the links are at least logarithmic.

Redefine the parameter $\nu$, by $$\nu \triangleq
n^{O(1)}.$$  This of course influences $\pmax$ and $k$
because $\pmax\triangleq \pmax\geq
  (\nu+2)\cdot \diam(G)$ and $k\eqdf \lceil \log (1+3\pmax) \rceil$.
However, in this setting $\pmax$ is polynomial in $n$ and $k=\Theta(\log n)$.

The following theorem shows that it is easy to achieve a logarithmic competitive
ratio if $B/c=n^{O(1)}$ and $B,c\geq k$.
\begin{theorem}\label{thm:largeBc}
  There exists an online deterministic algorithm for packet routing in
  $d$-dimensional grids with a competitive ratio of $O(\log n)$ if $B/c=n^{O(1)}$,
  and $B,c\geq k$. In this algorithm, packets are either rejected or routed but not
  preempted.
\end{theorem}

\begin{proof}
  Scale $B$ and $c$ by setting $B' \gets \lfloor{\frac Bk}\rfloor$ and $c' \gets
  \lfloor{\frac ck}\rfloor $.  Run the \IPP\ algorithm over the space-time graph
  $G^{st}$ with the scaled capacities $B'$ and $c'$ to decide which requests are
  rejected and which are routed. We claim that the routes computed by the \IPP\
  algorithm are a valid routing. Indeed, \IPP\ is $(2,k)$-competitive with respect to
  $B'$ and $c'$. Hence, the same packing of paths is $(O(k),1)$-competitive
  with respect to $B$ and $c$.
  The theorem follows since $k=O(\log n)$.
\end{proof}

\section{A Randomized Algorithm for the One Dimensional Case}
\label{sec:randalg} In this section we design and analyze a
randomized algorithm for routing packets in uni-directional
line networks. Our randomized algorithm achieves a
competitive ratio of $O(\log n)$.

The randomized algorithm applies only to the setting in which requests are {without} deadlines (i.e., $d_i = \infty$),
hence each packet is specified by a $3$-tuple $r_i=(a_i,b_i,t_i)$.

The randomized algorithm deals with all values of buffer sizes and communication link
capacities in the range $[1,O(\log n)]$. We do not require that $B,c\geq 3$ as in the
deterministic algorithm.

In particular, it holds also for unit buffers.  In
Sec.~\ref{sec:preprocess}-\ref{sec:together} we deal with
the case that both $B$ and $c$ are in $[1,\log n]$. We
consider this case to be the most interesting one. In
Sec.~\ref{sec:largeB} we deal with the case of $\log n\leq
B/c\leq n^{O(1)}$.  In Sec.~\ref{sec:smallBlargec} we deal
with the case of $B \in [1,\log n]$ and $c \in [\log n,
\infty)$.

\renewcommand{\arraystretch}{2}
\begin{table}[H]
\begin{centering}
\begin{tabular}{|c|c|c|}
\hline
$B$ & $c$ & Sections\tabularnewline
\hline
\hline
$[1,\log n]$ & $[1,\log n]$ & \ref{sec:preprocess}-\ref{sec:together} \tabularnewline
\hline
$[\log n, \infty)$ & $\left[\lceil \frac{B}{n^{O(1)}}\rceil, \frac{B}{\log n}\right]$ & \ref{sec:largeB}\tabularnewline
\hline
$[1,\log n]$ & $[\log n, \infty)$ & \ref{sec:smallBlargec}\tabularnewline
\hline
\end{tabular}
\par\end{centering}
\caption{Values of $B$ and $c$ in which our algorithm achieves logarithmic competitive ratio. In particular, it holds also for unit buffers, i.e., $B=1$. We consider the first case to be the most interesting one.}
\label{table:discussion}
\end{table}
\renewcommand{\arraystretch}{1}
\subsection{Outline of Modifications}\label{sec:comparison}
Our goal is to reduce the  $O(\log^5 n)$ competitive ratio of the deterministic algorithm (see
Theorem~\ref{thm:alg}) to a logarithmic competitive ratio with
the help of randomization.  In this section we outline the techniques that are
employed to achieve this goal.

In the randomized algorithm, the online integral packing algorithm is
applied to the sketch graph (without downscaling of capacities). To
simplify the discussion assume that $B=c=1$.  Since the load on every
edge in the sketch graph is at most $k$, and $k$ also equals the length of
the tile side, this implies that $O(k^2)$ paths traverse each tile side.

The ratio between the area and the perimeter of a tile is $\Theta(k)$.
As the number of requests that start in a tile is proportional to the
area of a tile, and the number of requests that can enter or exit a tile
is proportional to the perimeter of a tile, we need to avoid losing a
factor of $\Theta(k)$ in the competitive ratio. We do this by
\emph{randomly sparsifying} the requests. The goal of this
sparsification is to leave a $\Theta(1/k)$ fraction of the requests so
that a constant fraction of the remaining requests can be routed out
of their starting tile.

To facilitate detailed routing, we consider three (non-disjoint) areas within each
tile: (1)~a part in which new requests may start, (2)~a part dedicated to routing,
and (3)~a part in which requests reach their destination.  The tiles are randomly
shifted so that a constant fraction of the requests ``agree'' with the designated
parts in the tiles.

Detailed routing of requests not rejected by the \IPP\
algorithm or by random sparsification is simpler and always
succeeds.

\subsection{Preliminaries}\label{sect:prelimline}

\paragraph{Tiling.}
The untilted space-time graph $G^{st}$ is partitioned into rectangular
tiles.  We denote length of each tile by $\hl$ and the height by $\vl$
(we also require that $\hl$ and $\vl$ are even).  Note that tiles may
not be squares as in the deterministic algorithm.  Dummy nodes are
added to the space-time graph $G^{st}$ so that all the tiles are
complete.

\paragraph{Random Shifting.}
The tiling is specified by two additional parameters $\phi_{\hl}\in [0,(\hl-1)]$ and $\phi_{\vl}\in [0,(\vl-1)]$, called the \emph{phase shifts}.  The phase shifts determine the position of the ``first'' rectangle; namely, the node $(\phi_{\hl},\phi_{\vl})$ is the bottom left corner of the first rectangle.

\medskip\noindent
Recall that the sketch graph has a node for every tile in the space-time graph (see
Section~\ref{sect:sketchgraph}).
Each horizontal edge has a capacity of $\vl\cdot B$, and
each vertical edge has a capacity of $\hl \cdot c$,

\paragraph{Near and Far Requests.}
A request $r_i=(a_i,b_i,t_i)$ is classified as a \emph{near request} if the tile that
contains $(a_i,t_i)$ also contains a copy of $b_i$ (namely, the tile contains a
vertex $(b_i,t')$ for some $t'$). A request that is not a near request is classified
as a \emph{far request}. We denote the set of near and far requests by \near\ and
\far, respectively.

A routing of a request $r_i\in \far$ cannot be confined to a single tile. A
routing of a request $r_i \in \near$ may be within a tile or may span more than one
tile (our algorithm attempts to route near requests only within a single tile).

\paragraph{SW-Far requests.}
We partition each tile of the untilted space-time graph into four ``quadrants'' as depicted
Fig.~\ref{fig:quad}.
\begin{figure}[h]
  \centering
    \includegraphics[width=0.3\textwidth]{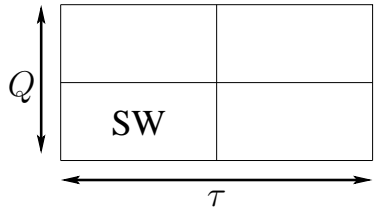}
  \caption{The south-west (SW) quadrant of a tile.}
\label{fig:quad}
\end{figure}

The tiling and random shifting defines the following random subset of the requests. Let
$R^+\subseteq R$ denote the subset of requests whose source vertex is in SW-quadrant of a tile.
The subset $\far^+$ is defined by
\begin{align*}
  \far^+ &\triangleq R^+ \cap \far.
\end{align*}

\paragraph{Online Integral Packing of Paths of Far Requests.}
The \IPP\ algorithm is applied only to $\far^+$ requests
over the sketch graph $S$ (see Line~\ref{line:IPP} in
Algorithm~\ref{alg:algRand}).

\paragraph{Multiple Simultaneous Requests from The Same Node.}
If multiple requests arrive simultaneously to the same
node, then even the optimal routing can serve at most $c+B$
packets among these packets. Since this limitation is
imposed on the optimal solution, the path packing algorithm
can abide this limitation as well without decreasing its
competitiveness. The online algorithm chooses $c+B$ packets
whose destination is closest to the source node, as
formalized in the following proposition.

\begin{proposition}\label{prop:filter}
    W.l.o.g. each node injects at most the closest $c+B$ requests at each time step.
\end{proposition}

\subsection{Randomized Algorithm: Preprocessing}\label{sec:preprocess}

\paragraph{Tiling parameters.}
The tile side lengths are set so that the trivial greedy routing algorithm is $O(\log
n)$-competitive for requests classified as near. Each tile has length $\hl$ and
height $\vl$.
Recall that $B,c\leq \log n$.
\begin{defn}\label{def:xy}
    \begin{enumerate}[(i)]
    \item
    If $B\cdot c < \log n$, then
    $\hl=2\lceil (\log n)/c \rceil$ and $\vl=2\cdot \lceil (\log n)/B\rceil$.
    \item
    If $B\cdot c \geq \log n$, then
     $\hl=2B$ and $\vl=2c$.
    \end{enumerate}
\end{defn}

\begin{proposition}\label{prop:tiling}
The choice of the tiling parameters implies the following:
  \begin{enumerate}
  \item $\hl+\vl = O(\log n)$.
  \item The capacity of each sketch edge is at least $\log n$.
  \item The ratio of maximum capacity to minimum capacity in the sketch graph is bounded by $2$.
 \end{enumerate}
\end{proposition}

\begin{proof}
  The first part of the proposition follows from the assumption that $B,c\in[1,\log
  n]$.  The capacity $c(e)$ of a horizontal edge $e$ in the sketch graph is $\vl\cdot
  B$.
If $Bc\geq \log n$, then $c(e)= 2Bc > \log n$ and all the sketch edges have the same capacity.
If $Bc< \log n$, then $c(e)\geq 2\frac{\log n}{B} \cdot B = 2\log n$.
Moreover, the ratio of maximum capacity to minimum capacity is bounded
by $2$. Indeed,
      \begin{eqnarray*}
        \frac {\vl \cdot B}{\hl \cdot c} & \leq & \frac {2 \cdot (1+\log n /B)\cdot B} {2\cdot (\log n/c) \cdot c}\\
        & = & \frac {\log n + B}{\log n} \leq 2\:.
    \end{eqnarray*}
Similarly, the ratio $\frac {\hl c}{\vl B} \leq 2$, and the proposition follows.
\end{proof}

To simplify the presentation, we assume that $\hl c=\vl B$
(we can obtain this by reducing the capacities by a factor
of at most $2$, which affects the competitive ratio only by
a factor of $2$). Let $c^S$ denote the capacity of the
sketch edges to the neighboring tiles.

\begin{proposition} \label{prop:class}
  If the phase shifts $\phi_{\hl}$ and $\phi_{\vl}$ are
  chosen independently and uniformly at random, then $E(|\opt (R^+)|)
  = \frac 14 \cdot |\opt(R)|$.  By a reverse Markov inequality, $$\Pr \left[|\opt
  (R^+)| \geq \frac 18 \cdot |\opt(R)|\right] \geq \frac 17.$$
\end{proposition}

\begin{proof}
  Since the phase shifts $\phi_{\hl}$ and $\phi_{\vl}$ are independent
  and uniformly distributed, the probability that a request $r_i\in R$
  is also in $R^+$ is $1/4$.  By linearity of expectation, $E(|\opt
  (R^+)|)= \frac 14 \cdot |\opt(R)|$.

  Plugging $X=|\opt (R^+)|$, $d=\frac 18 \cdot |\opt(R)|$ and
  $a=|\opt(R)|$ in Lemma~\ref{lemma:revMarkov} (See
  Appendix~\ref{sec:RevMarkovproof}) yields the second part of the
  proposition, i.e., $\Pr \left[|\opt(R^+)| \geq \frac 18 \cdot
    |\opt(R)|\right] \geq \frac 17$.
\end{proof}

\subsection{Algorithm for Requests in \far$^+$}\label{sec:far}
In this section we present an online algorithm for the
requests in the subset $\far^+$. Similarly to the
deterministic algorithm in Section~\ref{sec:outline}, the
$\far^+$-Algorithm invokes the \route\ algorithm (in
Step~\ref{line:IPP}) and applies detailed routing (in
Step~\ref{line:I}). The additional randomized steps are
employed in Step~\ref{line:toss}, and Step~\ref{line:load}.
Note that randomized algorithm is non-preemptive, that is,
if a packet is not rejected then it is guaranteed to arrive
to its destination.

\subsubsection{Description of The $\far^+$-Algorithm}
\paragraph{Parameters.}
Set the maximal path length in the sketch graph to be $\pmax\eqdf 4n$. We set
the probability $\lambda$ of the biased coin in step~\ref{line:toss} of
$\algf$ to be $\lambda=1/(200k)$, where $k=\lceil \log (1+3\pmax) \rceil$.

\begin{algorithm}
    \textbf{Upon arrival} of a packet request $r_i = (a_i,b_i,t_i)$ in $\far^+$
    proceeds as follows (if $r_i$ is rejected in any step, then the algorithm does not continue with the next steps):
\begin{enumerate}
\item
\label{line:IPP}

Reduce the packet requests to an online integral path
packing over the sketch graph with paths of length at
most $\pmax$. Execute the \route\ algorithm with respect
to these path requests. If the path request is rejected
by the \route\ algorithm then \textbf{reject} $r_i$.
Otherwise, let $\hat p_i$ denote the sketch path assigned
to request $r_i$.

\item \label{line:toss} Toss a biased $0$-$1$ coin $X_i$ such that
  $\Pr (X_i=1)=\lambda$. If $X_i=0$, then \textbf{reject} $r_i$.
\item
\label{line:load}\label{item:quarter}
If the addition of $\hat p_i$ causes the load of any sketch edge to be at least
$1/4$, then \textbf{reject} $r_i$.
\item\label{line:I} Apply $I$-routing to $r_i$.  If $I$-routing fails,
  then \textbf{reject} $r_i$. Otherwise, \textbf{inject} $r_i$ with
  the sketch path $\hat p_i$ and alternate between $T$-routing and
  $X$-routing.
\end{enumerate}
\caption{The $\far^+$-Algorithm. The input to the algorithm is a sequence of packet requests in $\far^+$ and it either rejects or injects.}\label{alg:algRand}
\end{algorithm}

The listing of the randomized algorithm appears in
Algorithm~\ref{alg:algRand}.
The input to the algorithm is the sequence of requests in $\far^+$ which is processed as follows:
\begin{inparaenum}[(1)]
\item The \route\ algorithm computes an integral packing of paths over the sketch graph $S$
  under the constraint that the length of a path is at most
  $\pmax$. In Proposition~\ref{lemma:nB}, we show that this
  constraint reduces the optimal fractional throughput by a factor of
  at most two. Algorithm \IPP\ remembers all accepted requests, even
  those that are rejected in subsequent steps. By
  Theorem~\ref{thm:IPP}, the computed paths constitute an
  $(O(1),k)$-competitive packing, for $k=O(\log n)$.
\item The probability $\lambda$ is set to $\frac{1}{\Theta (k)}$.
\item We maintain the invariant that after line~\ref{line:load}, the
  load of every sketch edge is at most $1/4$.
\item $I$-routing deals with routing the request out of the initial
  SW-quadrant and is described in Sec.~\ref{sec:detail}.  The rest of
  the path is computed based on the sketch path $\hat p_i$. This
  computation is performed locally and on-the-fly by alternating between two routing
  algorithms called $T$-routing and $X$-routing (described in
  Section~\ref{sec:detail}).
\end{inparaenum}

\paragraph{Remark.}
One may consider applying random sparsification before the \route\ algorithm is
  invoked. The motivation for such a variation is to avoid congesting the network
  with requests destined to be rejected. Apart from reducing the load of sketch
  edges, random sparsification facilitates successful $I$-routing (see
  Lemma~\ref{lemma:sparse}). This means that sparsification needs to be applied after
  the online path packing algorithm.

\subsubsection{Detailed Routing}\label{sec:detail}
The \IPP\ Algorithm computes a sketch path $\hat p_i$. If we wish to
route the packet, we need to compute a path in $G^{st}$. We refer
to this path as the \emph{detailed path}.
Three routing algorithms are employed for computing different parts the detailed path (see Fig.~\ref{fig:detail}):
\begin{inparaenum}[(1)]
  \item $I$-routing: from $(a_i, t_i')$ to the north or east boundaries of the SW-quadrant.
  \item $T$-routing: deals with routing in the north-west quadrant (NW-quadrant) and  the south-east quadrant (SE-quadrant) of a tile.
  \item $X$-routing: $X$-routing deals with routing in the north-east quadrant (NE-quadrant).
\end{inparaenum}
Let $\algfar\subseteq R^+$ denote the subset of requests that were successfully routed by $I$-routing.
Let $p_i$ denote the detailed path of a request $r_i\in \algfar$.
The packing $\{p_j \mid r_j\in \algfar\}$ satisfies the following
invariants:
    \begin{figure}[h]
      \centering
        \includegraphics[width=0.35\textwidth]{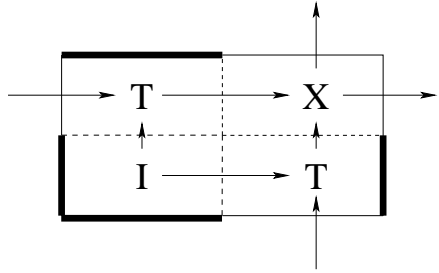}
      \caption{Allowed detailed routes in tile quadrants. Paths may not
        cross the thick lines.}
      \label{fig:detail}
    \end{figure}
\begin{enumerate}
\item The source of $p_j$ is in the SW-quadrant of a rectangle.
\item The prefix of $p_j$ till it exits the SW-quadrant is straight.
\item For every tile, $p_j$ may enter the tile only through the right half of the south side or the upper half of the west side.
\item For every tile, $p_j$ may exit the tile only through the right half of the north side  or the upper half of the east side.
\item Except for the first bend of $p_j$, every bend
    corresponds to a bend in the sketch path $\hat p_j$.
\item At most $c^S/4$ paths are routed out of the
    SW-quadrant.
\item The load of every edge in $G^{st}$ is at most one (i.e., all capacity constraints are satisfied).
\end{enumerate}

\paragraph{$I$-Routing.}
The goal of $I$-routing is simply to exit the SW-quadrant
either from its east side or its north side. $I$-routing
deals with routing paths that start in the SW-quadrant of a
tile till the north or east side of the SW-quadrant.
$I$-routing uses only straight paths.

By Proposition~\ref{prop:filter}, at most $B+c$ requests
are input at each node of $G^{st}$ to Algorithm \route.
These requests are ordered arbitrarily.  We therefore
consider each SW-quadrant as a three dimensional cube of
dimensions $\frac {\vl}{2} \times \frac {\hl}{2} \times
(B+c)$ where each node in the quadrant has $B+c$ copies.
The $i$th request that arrives to node $(v,t)$ is input to
node $(v,t,i)$ in the cube.  We refer to each copy of the
quadrant in the cube as a \emph{plane}.  Namely, the $i$th
plane is the set of nodes $(v,t,i)$ in the cube.
$I$-routing deals with each $\frac {\vl}{2} \times \frac
{\hl}{2}$ plane separately,

$I$-routing tries to route horizontally the first $B$ requests that
start at a node. Similarly, $I$-routing tries to vertically route the
 requests that arrive after that.  By trying to route a request,
we mean that if the corresponding row or column in the plane is free,
then the request is routed (and that row or column in the plane is
marked as occupied); otherwise the request is rejected.

Finally, we  limit the number of paths that emanate from each
side of the SW-quadrant by $c^S/4$, where $c^S$ denotes the capacity
of the sketch edges to the neighboring tiles. Thus after $c^S/4$ requests have been successfully
$I$-routed out of the SW-quadrant, all subsequent requests from this
SW-quadrants fail.

 Note that $I$-routing is computed
before the packet is injected and does not preempt packets (after they are injected) since
precedence is given to existing paths.

\paragraph{$T$-routing.}
The NW-quadrant and the SE-quadrant have a ``blocked'' side
that is depicted by a thick link in
Figure~\ref{fig:detail}. Paths may not traverse the blocked
side. $T$-routing deals with routing in these two
quadrants. Paths may enter these quadrants from two sides
but must exit through a third side (unless they reach a
copy of their destination).  We show that $T$-routing is
always successful (in fact, $T$-routing is similar to
detailed routing in internal segments described in
Sec.~\ref{sec:detailed internal}).

Consider a SE-quadrant: each path enters through the south or west sides of the
quadrant, and should be routed to the north side of the quadrant. The detailed paths
of south-to-north paths are simply vertical paths without bends (such paths are given
precedence). The detailed paths of west-to-north paths are obtained by traveling
eastward until a bend can be made, namely, the vertical path to the north side is not
saturated.  Since both path types contain at most $c^S/4$ paths, and since $c^S/2$
paths can cross the north side of the quadrant, $T$-routing never fails.

\paragraph{$X$-routing.}
$X$-routing deals with routing in the NE-quadrant. Note
that a path may enter the  NE-quadrant from its west side
or from its south side. Moreover, a path may exit the
NE-quadrant from its east or north side. We show that
$X$-routing is always successful.

$X$-routing is implemented by super-positioning two
instances of $T$-routing (see Fig.~\ref{fig:X-routing}). We
partition the traffic in a NE-quadrant to two parts based
on the side from which the path exits the quadrant. As in
$T$-routing, precedence is given to straight traffic.  A
bend takes place whenever a free path is available.
Clearly, a straight path is successfully routed. Paths that
needs to turn are blocked by at most $c^S/4$ paths from the
other part. There are at most $c^S/4$ paths that need to
turn, and the capacity of the side of the quadrant is
$c^S/2$, hence $X$-routing is always successful. (Note that
knock-knee bends are not required, although they could be
incorporated.)

    \begin{figure}[h]
      \centering
        \includegraphics[width=0.36\textwidth]{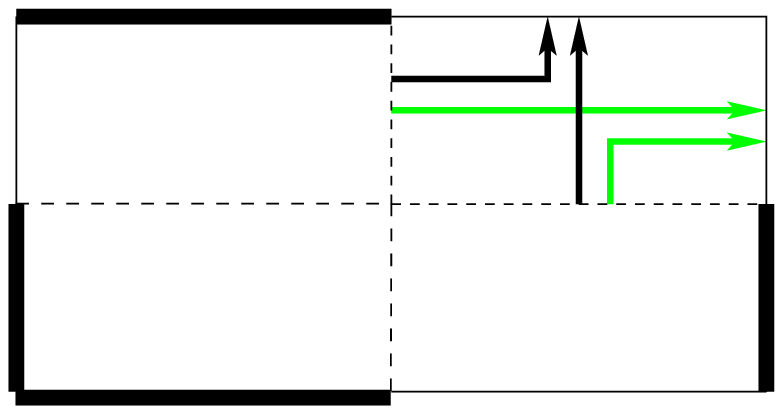}
      \caption{$X$-routing is implemented by super-positioning two instances of $T$-routing depicted by black and grey arrow.}
      \label{fig:X-routing}
    \end{figure}

\paragraph{Last Tile.}
%
Detailed routing in the last tile employs greedy shortest path routing.  If a packet
enters the last tile from the south side, then it simply continues north until it
reaches its destination. Note that no such packet may enter the last tile from the
west side. Indeed, if a sketch path enters $s$ from the west side and $s$ is the last
tile in the sketch path, then the neighboring tile from the west contains a copy of the
destination, and hence $s$ is not the last tile in the sketch path.

\subsubsection{Analysis}

\paragraph{Notation.}
We define the following chain subsets of requests
$$\algfar \subseteq \Rinj \subseteq \RIPPt \subseteq
\RIPP\subseteq \far^+\:,$$ as follows.  $\RIPP$ is the
subset of requests accepted by the $\IPP$ algorithm in
Line~\ref{line:IPP}. $\RIPPt\subseteq \RIPP$ is the subset
of requests for which the biased coin flip $X_i$ equals $1$
in Line~\ref{line:toss}.  $\Rinj \subseteq \RIPPt$ is the
subset of requests whose addition did not cause a sketch
edge to be at least $1/4$ loaded in Line~\ref{line:load}.
$\algfar\subseteq \Rinj$ is the subset of requests for
which detailed routing is successful in Line~\ref{line:I}
(recall, that only $I$-routing may fail).

Let $\opt_f(R)$ (respectively, $\opt(R)$) denote an optimal fractional
(respectively, integral) packing of paths in $G^{st}$ with respect to the
requests $R$. An optimal packing of paths in the space-time graph $G^{st}$ in
which the length of the paths in the packing is bounded by $\pmax^{st}$ is
denoted by $\opt_f(R\mid \pmax^{st})$.

The following theorem states that the invocation of the \IPP\ algorithm
assigns routes to a constant fraction of an optimal solution.
\begin{theorem}\label{thm:IPP rand}
  \begin{align*}
   | \IPP(\far^+ \mid \pmax) |& \geq \frac 14 \cdot| \opt(\far^+)|.
  \end{align*}
\end{theorem}
\begin{proof}
  The proof of the theorem is divided into three parts (summarized by
  Equations~\eqref{eq:part 1}-\eqref{eq:part 3}).  The first part states that a fractional
  packing is not smaller than an integral one.
  \begin{align}\label{eq:part 1}
    | \opt_f(\far^+)|& \geq | \opt(\far^+)|.
  \end{align}

  The second part shows that bounding the path lengths reduces the throughput only by
  a factor of $2$.
  \begin{lemma}[\protect{\cite[Claim 4.5]{AZ}}]\label{lemma:nBline} Let
    $\pmaxst\triangleq 2\cdot (n-1)\cdot (1+B/c)$.  Then,
    \begin{align}
      |\opt_f(\far^+ \mid \pmaxst)| \geq \frac 12 \cdot |\opt_f(\far^+)|.
      \label{eq:part 2}
    \end{align}\end{lemma}

  The third part shows that paths of length at most $\pmaxst$ in the space-time graph
  are mapped to paths of length at most $4n$ in the sketch graph.
  \begin{proposition}\label{prop:4n}
    Every path $p$ in $G^{st}$ of length at most $\pmaxst$ is mapped to a path $\hat
    p$ in the sketch graph $S$ of length at most $4n$. Hence, by the
    $(2,k)$-competitiveness of the \IPP\ Algorithm, it follows that:
    \begin{align}
      \label{eq:part 3}
      | \IPP (\far^+ \mid \pmax)| \geq \frac 12 \cdot |\opt_f(\far^+ \mid \pmaxst)|.
    \end{align}
  \end{proposition}

\begin{proof}
  Let $p$ denote a path of length at most $\pmaxst\eqdf 2\cdot (n-1)\cdot (1+B/c)$ in
  $G^{st}$. We partition the edges of $\hat p$ into horizontal edges and vertical
  edges in $\hat p$.  The number of vertical edges in $p$ is bounded is $n$ and the
  same holds also for $\hat p$.

  We now prove that the number of horizontal edges in $\hat p$ is at most $3n$. For
  every row $i$ in $G^{st}$, let $n_i$ denote the number of horizontal edges of $p$
  in the $i$th row.
  Similarly, for every row $i$ in the sketch graph, let $\hat{n}_i$
  denote the length of the intersection of $\hat p$ with the $i$th row of the sketch graph.  Let $[\alpha_{i}, \beta_{i}]$ denote the interval of rows of
  $G^{st}$ that are mapped to the $i$th row of the sketch graph (note that
  $\beta_i-\alpha_i$ is simply the height of a tile).

  By Def.~\ref{def:xy}, the length of every tile is at least $2B$. Indeed, if $B\cdot
  c> \log n$, then the length $\hl$ equals $2B$.  If $B\cdot c \leq \log n$, then the length
  $\hl \geq 2\log n/c \geq 2B$.  It follows that
  \begin{align*}
    \hat n_i &\leq \left\lceil{\frac{\sum_{j= \alpha_{ i}}^{\beta_{i}}
          n_j}{2B}}\right\rceil \leq \frac{1}{2B} \cdot \sum_{j= \alpha_{
        i}}^{\beta_{i}} n_j +1.
  \end{align*}
  Hence, $\sum_{i} \hat n_{i} \leq \frac{\pmax^{st}}{2B} + n \leq
  3n$. We conclude that the length of the path $\hat p$ is at most $4n$, as required.
\end{proof}
Equations~\eqref{eq:part 1}-~\eqref{eq:part 3} completes the proof of Theorem~\ref{thm:IPP rand}
\end{proof}

The following proposition shows that, in expectation  over the biased coins
tosses in Line~\ref{line:toss}, at most a quarter of the sketch paths are
rejected due to ``$\frac 14$-loaded''  edges in line~\ref{line:load} of the
$\far^+$-Algorithm.
\begin{lemma}\label{lemma:Rinj}
  If $n > 16$, then
$$E(|\Rinj|) \geq \frac 34 \cdot E(|\RIPPt|)\:.$$
\end{lemma}

\begin{proof}
The idea it to show that, after random sparsification, the load of every sketch edge is at
most $1/4$ with high probability. This implies that few requests are rejected
as a result of causing the load of an edge to be greater than $1/4$.

Let $\hat p_i$ denote the sketch path of $r_i$.  Given a sketch edge $e$, let $P(e)
\triangleq \{ \hat p_i : r_i\in \RIPP, e \in \hat p_i\}$ denote the set of sketch
paths that traverse $e$.  Similarly, let $P^{\lambda}(e) \triangleq \{ \hat p_i : r_i\in
\RIPPt, e \in \hat p_i\}$ denote the set of paths that traverse $e$ after random
sparsification.  We first claim that, for a constant $\gamma>200$, for $n > 24$, and
for every sketch edge $e$, 
\begin{align}
  \label{eq:1}
  \Pr \left(|P^{\lambda}(e)| > \frac{c(e)}{4} \right) < \frac 1 {16n}.
\end{align}
We now prove Equation~\eqref{eq:1}.  Since $\RIPP$ is $(2,k)$-competitive, it follows
that $$|P(e)| \leq k\cdot c(e)\:.$$ The tossing of the biased coins with $\lambda =
1/(\gamma k)$ with $\gamma=200$, implies that
\begin{align*}
  E (|P^{\lambda}(e)|) = \lambda \cdot |P(e)| \leq \lambda k \cdot c(e) = \frac {1}{\gamma}\cdot c(e).
\end{align*}

\noindent
The following sequence of equations is explained below.
    \begin{eqnarray*}
        \Pr \left(|P^{\lambda}(e)| > \frac{c(e)}{4} \right)
        & = & \Pr \left(|P^{\lambda}(e)| \geq (1+\delta)\frac{c(e)}{\gamma} \right) \\
        & < &    \left( \frac {e^{\delta}}{(1+\delta)^{(1+\delta)}} \right)^{\frac{c(e)}{\gamma}}\\
        & \leq & \left( \frac {e^{\delta/\gamma}}{(1+\delta)^{(1+\delta)/\gamma}} \right)^{2\cdot\log n} \\
        & = & \left(\frac {e^{\frac14 - \frac{1}{\gamma}}}{(\frac {\gamma}{4})^{\frac 14}} \right)^{2\cdot\log n},
    \end{eqnarray*}
    The first line holds if $\delta$ satisfies $\frac
    {1+\delta}{\gamma} = \frac {1}{4}$.
    The second line is due to a multiplicative Chernoff bound~\cite{MU}.
    The third line is implied by Proposition~\ref{prop:tiling} since
    $c(e) \geq 2\cdot\log n$.
The last line follows by the definition of $\delta$.

 Since $\gamma=200$ and $n > 16$, then $\left(\frac {e^{\frac14 - \frac{1}{\gamma}}}{(\frac
      {\gamma}{4})^{\frac 14}}\right)^2 < 2^{-2} < 2^{-\frac {\log 16n}{\log
        n}}$ and therefore, $\Pr \left(|P^{\lambda}(e)| > \frac{c(e)}{4} \right) < 2^{-\log 16n}$ and Equation~\eqref{eq:1}
    holds.

    Since $\pmax=4n$, the length of each sketch path is at most $4n$.  By
    Equation~\eqref{eq:1} and by applying a union bound it follows that
\begin{align*}
  \Pr \left(r_{i} \not\in \Rinj ~|~ r_i \in \RIPPt\right) & \leq \Pr \left( \exists~ e \in \hat p_i : P^{\lambda}(e) > \frac 14 \cdot c (e)\right) \\
  &\leq 4n \cdot \frac 1{16n} = \frac 14\:.
\end{align*}
The lemma follows by linearity of expectation.
\end{proof}

The following theorem states that, in expectation, a $1/\Theta(k)$ fraction
of the requests that are accepted by the \route\ algorithm are successfully routed.
\begin{theorem}\label{theorem:Rinjt}
  $E(|\algfar|)\geq \frac \lambda 4 \cdot |\RIPP|$.
\end{theorem}

\begin{proof}
  We first prove a Lemma and a Proposition. Lemma~\ref{lemma:sparse} deals with a
  projection of a random sparsification of a $0$-$1$ matrix. This
  lemma helps estimate the number of requests from $\RIPPt$ for which
  $I$-routing is successful in each plane (ignoring the effect of
  line~\ref{line:load} in the algorithm).
  Proposition~\ref{proposition:dom} helps analyze the effect of
  line~\ref{line:load} on the number of requests for which
  $I$-routing is successful.

\paragraph{Definitions.}

Let $I(\cdot)$ be an operator over $0$-$1$ matrices defined as follows. $I(X)$ is all
zeros except for the first nonzero entry in each row of $X$. Namely,
\[
I(X)_{i,j} \eqdf
\begin{cases}
  1 &\text{if $X_{i,j}=1$ and $X_{i,\ell}=0$, for every $\ell<j$}\\
0&\text{otherwise.}
\end{cases}
\]
The motivation for this definition is as follows. Suppose
that the matrix $X$ indicates the existence of packets in a
plane of a SW-quadrant in which packets are routed by
$I$-routing along rows out of the quadrant. The only
packets for which $I$-routing succeeds in this plane are
the packets that correspond to ones in $I(X)$.

Let $L\wedge B$ denote the matrix obtained by the coordinate-wise
conjunction of $L$ and $B$. For a matrix $X$, let $w(X)$ denote the number of
$1$'s in $X$.

In the following lemma we analyze the effect of random sparsification on $I$-routing along the rows of the SW-quadrant.
A similar effect occurs when considering $I$-routing along the columns of the SW-quadrant.
\begin{lemma}\label{lemma:sparse}
  Let $A$ and $Z$ be $0$-$1$ matrices whose dimensions are $\frac{\vl}{2} \times \frac{\hl}{2}$. Assume that the entries of $Z$ are i.i.d. $0$-$1$
  random variables with $E(z_{ij})=\lambda\:$.  Let $\lambda < \frac
  {2}{\hl}$. Then,
  $$E\Big(w\big(I(A \wedge Z)\big)\Big) \geq \frac{\lambda}{2} \cdot w(A)\:.$$
\end{lemma}
\begin{proof}
  Consider each row $A_i$ of $A$ and $Z_i$ of $Z$ separately.  The expectation of the
  $0$-$1$ random variable $w(I(A_i \wedge Z_i))$ equals the probability that it equals $1$.
  Note that
\begin{align*}
  \Pr(w(A_i\wedge Z_i)=0)&=(1-\lambda)^{w(A_i)}\\
&\leq e^{-\lambda \cdot w(A_i)}.
\end{align*}
Since $\lambda \cdot \tau/2 \leq 1$, it follows that $\lambda \cdot w(A_i)\leq 1$,
and hence
\begin{align*}
  \Pr(w(A_i\wedge Z_i)=1)&\geq 1- e^{-\lambda \cdot w(A_i)}\\
&\geq \frac{\lambda}{2} \cdot w(A_i).
\end{align*}
The lemma follows by linearity of expectation.
\end{proof}

We now return to the proof of Theorem~\ref{theorem:Rinjt}.  For every tile
consider its SW-quadrant as a three dimensional cube of dimensions
$\frac {\hl}{2} \times \frac {\vl}{2} \times (B+c)$. Recall that $I$-routing deals with
each $\frac {\hl}{2} \times \frac {\vl}{2}$ plane separately.

The lengths $\hl$ and $\vl$ of each tile are at most $2\log n$.  Recall that
$\lambda=\frac {1}{\gamma \cdot k}$ where $k\geq \log (1+3\cdot 4n)$. Hence,
if $\gamma=200$, then $1/\lambda = \gamma \cdot k \geq \frac{\tau}{2}$.

Assume that we skip Step~\ref{item:quarter} of the algorithm (namely, we do not check
that the load is bounded by $1/4$), and apply directly $I$-routing to the requests in
$\RIPPt$. Let $I_{\IPP^{\lambda}}$ denote the set $\{r_i\in \RIPPt : \text{$I$-routing
  succeeds in routing $r_i$}\}$.  We consider each of the $(B+c)$ planes separately, and by
Lemma~\ref{lemma:sparse} and linearity of expectation, we obtain
\begin{align}\label{eq:A}
    \nonumber
  E(|I_{\IPP^{\lambda}}|) &\geq \frac \lambda 2 \cdot |\RIPP|\\
&= \frac 12 \cdot E(|\RIPPt|)\:.
\end{align}

Furthermore, Lemma~\ref{lemma:Rinj} implies that:
\begin{align}\label{eq:B}
  E(|\RIPPt \setminus \Rinj|) &\leq \frac 14 \cdot E(|\RIPPt|)\:.
\end{align}

Hence,
\begin{align}\label{eq:almostqed}
  E(|I_{\IPP^{\lambda}}|)-E(|\RIPPt \setminus \Rinj|) &\geq \frac 14 \cdot E(|\RIPPt|)\:.
\end{align}

\paragraph{Notations.}
For a $0$-$1$ matrix $L$, let $\bar{L}$ denote negated matrix
$\bar{L}_{i,j}\eqdf1-L_{i,j}$.  For matrices $L$ and $B$, let $L\leq B$
denote $L_{i,j}\leq B_{i,j}$, for every $i$ and $j$.

\begin{proposition}\label{proposition:dom}
If $L\leq B$  then:
  \begin{eqnarray*}
  w(I(L))&\geq w(I(B))-w(B \wedge \bar{L}).
  \end{eqnarray*}
\end{proposition}

\begin{proof}
  It suffices to deal with each row separately. Let $B_i$ denote the $i$th row
  of the matrix $B$. We claim that if $w(I(B_i))=1$, then $w(I(L_i))=1$ or
  $w(B_i\wedge \bar{L}_i)\geq 1$.  Indeed, assume that $w(I(B_i))=1$ and
  $w(I(L_i))=0$. Then, $L_i$ is all zeros. Hence, $B_i \wedge \bar{L}_i = B_i$, and
  the proposition follows.
\end{proof}

\noindent
We now prove the following lemma.
\begin{lemma}\label{lemma:combine}
    For every outcome of the random biased coins:
  $$|\algfar| \geq |I_{\IPP^{\lambda}}| - |\RIPPt \setminus \Rinj|\:.$$
\end{lemma}
\begin{proof}
  Consider a specific tile $s$ and its SW-quadrant.  Fix an $i$-plane used by
  $I$-routing. W.l.o.g this $i$-plane corresponds to a horizontal $I$-routing. Define
  three $0$-$1$ matrices $A,Z$ and $L$ with dimensions $\frac{\vl}{2}\times \frac{\hl}{2}$, as
  follows:
  \begin{enumerate}
  \item Let $A$ be the matrix whose entries indicate the existence of a request $r\in
    \RIPP$ whose source vertex is in the $i$th plane of the SW-quadrant of the tile
    $s$. Namely, $A_{v,t}=1$ iff node $(v,t)$ receives at least $i$ requests in
    $\RIPP$.
  \item Let $Z$ denote a random matrix in which the entries are i.i.d. Bernoulli
    random variables with $Pr(Z_{v,t}=1)=\lambda$. These Bernoulli random variables
    correspond to the outcomes of the biased coin tosses in Step~\ref{line:toss} of
    the algorithm.
  \item Let $L$ be the matrix whose entries indicate the existence of a request $r\in
    \Rinj$ whose source vertex is in the $i$th plane of the SW-quadrant of the
    tile $s$.
  \end{enumerate}

  For a subset $W$ of requests, a tile $s$, and a plane index $i$, let
  $W(s,i)\subseteq W$ denote the subset of requests in $W$ whose source vertex is in
  the $i$th plane of the tile $s$.  Let $\bar L$ denote the negation of $L$.  By
  definition the following identities hold:
  \begin{enumerate}[(i)]
    \item $|\algfar(s,i)| = w(I(L))$,
    \item $|\RIPPt(s,i)|= w(A\wedge Z)$,
    \item $|I_{\IPP^{\lambda}}(s,i)| = w(I(A \wedge Z))$,
    \item $|\RIPPt(s,i) \setminus \Rinj| = w(A \wedge Z \wedge \bar{L})$.
  \end{enumerate}

It suffices to prove that
\begin{align}
  \label{eq:suffice}
  w(I(L))&\geq w(I(A\wedge Z))-w(A\wedge Z \wedge \bar{L}).
\end{align}
Since $L \leq (A \wedge Z)$, Equation~\eqref{eq:suffice} follows from Proposition~\ref{proposition:dom}, and the
lemma follows.
\end{proof}

We now complete the proof of Theorem~\ref{theorem:Rinjt}. By
Lemma~\ref{lemma:combine} and Equation~\eqref{eq:almostqed}, it follows that
$E(|\algfar|)\geq \frac 14 \cdot E(|\RIPPt|)$. Theorem~\ref{theorem:Rinjt}
follows since $E(|\RIPPt|)= \lambda \cdot |\RIPP|$.
\end{proof}

\begin{theorem}\label{thm:far}
  $E(|\algf|) \geq \Omega( \frac{1}{\log n})\cdot |\opt (\far^+)|$.
  \end{theorem}
  \begin{proof}
 By Theorem~\ref{theorem:Rinjt}, it follows that
    $E(|\algf|) \geq \Omega(\lambda) \cdot |\RIPP|$. By
    Theorem~\ref{thm:IPP rand}, $|\RIPP|\geq
    \Omega(|\opt(\far^+)|)$. The theorem follows since $\lambda=1/\Theta(k)=1/\Theta(\log
    n)$.
  \end{proof}
\subsection{Algorithm for Requests in \near}\label{sec:near}
In this section we present an online algorithm for the requests in the
subset $\near$.  The algorithm is a straightforward greedy vertical
routing algorithm.
Given a request $r_i\in \near$, the algorithm
attempts to routs the request vertically.

We emphasize that an optimal routing is not restricted to routing a
request $r_i\in \near$ within the tile.

\paragraph{Notations.}
Let $\algn$ denote the set of requests successfully routed by the \near-Algorithm
with respect to the requests in $\near$.  Let $\algn (s)$ denote the set of requests
routed by the \near-Algorithm within the tile $s$. Let $\near_s$ denote the set of
requests in $\near$ whose starting node is in the tile $s$. We abuse notation and
refer to the set of routed packets in an optimal routing with respect to $\near_s$\
also by $|\opt(\near_s)|$.

\begin{theorem} \label{thm:near}
For every tile $s$,
  $|\algn(s)| \geq \Omega(\frac {1}{\log n}) \cdot |\opt({\near}_s)|$.
\end{theorem}


\begin{proof}
    It suffices to prove that
\begin{align*}
  |\algn(s)| > \Omega \left(\frac {1}{\log n}\right) \cdot | \opt(\near_s)\setminus \algn(s)|
\end{align*}
%
We consider a bipartite conflict graph between requests in $\algn(s)$
and $\opt(\near_s)\setminus \algn(s)$. There is an edge $(r,r')\in \algn(s)\times \opt(\near_s)\setminus \algn(s)$ if
the vertical path of $r$ shares an edge with the path of $r'$ in $\opt(\near_s)\setminus \algn(s)$.

Since at most $c$ requests  can traverse the same vertical
edge, it follows that a route of a request in $\algn(s)$
conflicts with at most
\begin{align*}
  deg(r) & \leq \vl\cdot c\:.
\end{align*}

If $r' \not \in \algn(s)$, then it either encountered a saturated
horizontal edge or a saturated vertical edge. Hence, the degree of
$r'\in \opt(\near_s)\setminus \algn(s)$ is at least
\begin{align*}
  deg(r') &\geq c\:.
\end{align*}
By counting edges on each side we conclude that
\begin{align*}
  \frac{ | \opt(\near_s)\setminus \algn(s)|}{ |\algn(s)|} & \leq
  \frac{\max deg(r)}{\min deg(r')}\\
  &\leq \frac{\vl\cdot c}{c}.
\end{align*}
By Definition~\ref{def:xy}, $\vl \leq 2 \cdot \log n$, and
the theorem follows.
\end{proof}

\begin{coro}\label{coro:near}
$|\algn| \geq \Omega\left( \frac{1}{\log n}\right)\cdot |\opt (\near)|$.
\end{coro}
\subsection{Putting Things Together}\label{sec:together}
The online randomized algorithm \alg\ for packet routing on a directed
line proceeds as follows.
\begin{enumerate}
\item Choose the tiling parameters $\hl,\vl$ according to Definition~\ref{def:xy}.
\item Choose the phase shifts $\phi_{\hl}\in[0,\hl-1],\phi_{\vl}\in[0,\vl-1]$ of
  tiling independently and uniformly at random.
\item Flip a random fair coin $b\in \{0,1\}$.
\item If $b=1$, then consider only requests in $\far^+$, and apply the
  $\far^+$-algorithm to these requests.
\item If $b=0$, then consider only requests in $\near$, and apply the
  $\near$-algorithm to these requests.
\end{enumerate}

\begin{theorem}\label{thm:algrand}
  If $B,c \in [1,\log n]$, then the competitive ratio of \alg\ is $O(\log n)$.
\end{theorem}
\begin{proof sketch}{Theorem~\ref{thm:algrand}}
    The chosen tiling parameters and phase shifts induce a classification of the requests to two classes: $\near$ and $\far+$.
    With probability $\frac 12$ the random fair coin $b$ chooses the bigger class.
    Theorem~\ref{thm:far} and Corollary~\ref{coro:near} state that $\algf$ and $\algn$ are $O(\log n)$ competitive, and the theorem follows.
\end{proof sketch}

\subsection{Large Buffers} \label{sec:largeB}

In this section we consider a special setting in which the buffers are large.
Note that the Algorithm fails if $B=\omega(\log n)$ both with near and far requests.
Formally, assume that $\log n \leq B/c \leq n^{O(1)}$.

We briefly mention the required modifications.  The tiling parameters are $\hl=B/c$
and $\vl=1$.  This implies that there are no near requests and all requests are
classified as far.  Each tiles is partitioned in to a left half and a right half.
The algorithm considers only requests whose source vertex is in the left half of a
tile; such requests are denoted by $R^+$. Note that random shifting is employed so
that on the average $R^+$ contains half the requests.

The north and south side of the left half of each tile are ``blocked'' so that
detailed routing does not traverse these sides.  This means that $I$-routing is only
along horizontal edges.  In the right half of each tile, three $T$-routing are super
imposed. The first $T$-routing is for the paths that enter the tile from the west
side. These paths traverse the left half horizontally and then in the right half
undergo $T$-routing (so that they exit from the east or north side of the right
half).  The second $T$-routing is for the paths that enter the tile from the south
side of the right half. Finally, the third $T$-routing is for continuing the paths of
the $I$-routing from the border between the halves to the north and east sides of
the right half of the tile.

Path lengths are bounded as before (this is why we require
that $B/c$ is polynomial).  In addition the random
sparsification parameter $\lambda$ is the same.

The algorithm proceeds as follows:
\begin{enumerate}
\item Execute the \route\ algorithm with respect to the path requests in $R^+$ over
  the sketch graph.

\item \label{line:toss 2} Toss a biased $0$-$1$ coin $X_i$ such that $\Pr
  (X_i=1)=\lambda$. If $X_i=0$, then \textbf{reject} $r_i$.

\item
\label{line:load 2}\label{item:quarter 2}
If the addition of $\hat p_i$ causes the load of any sketch edge to be at least
$1/4$, then \textbf{reject} $r_i$.

\item\label{line:I 2} Apply $I$-routing to $r_i$.  If $I$-routing fails, then
  \textbf{reject} $r_i$. Otherwise, \textbf{inject} $r_i$ with the sketch path $\hat
  p_i$ and apply $T$-routing till the destination is reached.
\end{enumerate}

In this setting, the ratio between the capacity of the sketch edges that emanate from
a tile to the number of requests whose source vertex is in the tile is constant. This
constant ratio simplifies the proof of the following theorem compared to the proof of
Theorem~\ref{thm:algrand}.

\begin{theorem}
  If $\log n\leq B/c\leq n^{O(1)}$, then there exists a randomized online algorithm
  that achieves a logarithmic competitive ratio for packet routing in a
  uni-directional line.
\end{theorem}

Recall that for the case where $B,c \in [\Omega(\log
n),\infty)$ and $B/c = n^{O(1)}$, there is an even simpler
and \emph{deterministic} online algorithm with $O(\log n)$
competitive ratio, as stated in Theorem~\ref{thm:largeBc}.

\subsection{Small Buffers \& Large Link Capacities}\label{sec:smallBlargec}
The case $B\in [1,\log n]$ and $c \in [\log n, \infty)$ is dealt with
by simplifying the algorithm. We briefly mention the required
modifications. The tile size is $\hl=1$ and $\vl=\log n/B$.  The maximum
path length is set to $2(n-1)(1+B/c)$ which is polynomial (i.e.,
tiling is not needed to reduce the path length).  Instead of
partitioning a tile into quadrants, we partition each tile into an
upper half and a lower half. The set $R^+$ is defined to the set of
requests whose origin is in the lower half of a tile.

The set $\near$ is dealt by a vertical path. Since in every
tile $s$, $|\algn(s)| \geq \min \{c,|\opt(\near_s)|\}$ and
since $|\opt(\near_s)| \leq \frac {\log n }{B} \cdot (B +
c)$, it follows that $\frac{|\algn(s)|}{|\opt(\near_s)|}
\geq \frac {1}{\log n}$.

The set $\far^+$ is dealt by invoking a variation of the
$\far^+$-Algorithm. The modified invariants for detailed
routing are that paths may not enter or exit horizontally
through the lower half of a tile (but, of course, may
traverse the tile vertically).  $I$-routing simply routes
the first $\frac 34 \cdot c$ requests vertically. The
remaining capacity of $\frac c4$ is reserved for incoming
paths from the south side. In the upper half of each tile,
$X$-routing on a single column is employed.

We conclude with the following theorem.
\begin{theorem}
  If $B\in [1,\log n]$ and $c \in [\log n, \infty)$, then there exists a randomized
  online algorithm that achieves a logarithmic competitive ratio for packet routing
  in a uni-directional line.
\end{theorem}

\paragraph{Remark.}
The space-time graph seems to assign symmetric roles to the time axis and the space
axis.  Such a symmetry would imply that one could reduce the case of large buffers to
the case of large link capacities. However, this is not true due to the definition of
a destination.  A destination (in the space-time graph) is a row of
vertices (namely, the set of copies of an original vertex).
This implies that one cannot simply transpose the graph and exchange the roles of
space and time.

\section{Open Problems}
Two basic problems related to the design and analysis of
online packet routing remain open even for uni-directional
lines.
\begin{inparaenum}[(i)]
\item Achieve a constant competitive ratio or prove a lower
    bound that rules out a
  constant competitive ratio.
\item Achieve a logarithmic competitive ratio by a distributed algorithm (as opposed
  to a centralized algorithm).
\end{inparaenum}

\subsection*{Acknowledgments}
We thank  Niv Buchbinder, Boaz Patt-Shamir and Adi Ros{\'e}n for useful discussions.

\bibliographystyle{alpha}
\bibliography{bib_sicomp}

\appendix

\section{Proof of Lemma~\ref{lemma:nB}}\label{sec:proofnB}

The following lemma shows that bounding path lengths in a fractional
path packing problem over a space-time graph to a polynomial
length decreases the fractional throughput only by a constant factor.  The
lemma is an extension of a similar lemma from~\cite{AZ}.

Consider a directed graph $G=(V,E)$ with edge capacities $c(e)$ and
buffer size $B$ in each vertex.  Let $G^{st}$ denote the space-time
graph of $G$ (see Section~\ref{sec:spacetime}).  Let $c_{\min} =
\min\{c(e) \mid e\in E\}$.  Let
$\dist_G(u,v)$ denote the length of a shortest path from $u$ to $v$ in
$G$. Let $\diam(G)$ denote the diameter of $G$ defined as follows
\begin{align*}
  \diam(G) &\eqdf \max \{ \dist_{G}(u,v) \mid \text{there is a path from $u$ to $v$ in $G$}\}.
\end{align*}

\textbf{Lemma~\ref{lemma:nB}} \emph{Let $\alpha\eqdf\frac{c_{\min}}{2\cdot
    (\sum_{e\in E} c(e)+n\cdot B)}$, $\nu\eqdf1/\alpha$, and $\pmax\geq
  (\nu+2)\cdot \diam(G)$.  Then, $$|\opt_f(R \mid \pmax)| \geq \frac
  1{2} \cdot \left(1-\frac 1e\right) \cdot |\opt_f(R)|\:.$$}

\begin{proof}
  Let $f^*$ denote an optimal fractional path packing in $G^{st}$ with respect to a set of flow
  requests $R$, that is $f^* = \opt_f(R)$.   The flow $f^*/2$ has a
  throughput that is half the throughput of $f^*$ and the load of each edge
  is at most $\frac {1}{2}$.

  Consider the following pipelining scheme. The time dimension is partitioned into
  intervals at multiples of $\diam(G)$. Let $\cut_i$ denote the set of edges in $G^{st}$ defined by
  \begin{align*}
    \cut_i &\triangleq \{ ((u,t),(v,t+1)) \mid u=v \text { or } (u,v)\in E, t=i\cdot \diam(G)\}.
  \end{align*}
  The fractional flow we construct is the sum of two flows $g$ and $h$
  defined as follows.  The flow $g$ is based on the flow $f^*/2$,
  where flow paths that traverse $\cut_i$, for $i>0$, are modified as
  follows:
  \begin{inparaenum}[(i)]
  \item If a flow path $p_j$ in $g$ traverses $\cut_i$, then split the
    path $p_j$ as follows: keep a fraction $(1-\alpha)$ of $p_j$ in
    the flow $g$ and transfer an $\alpha$-fraction of $p_j$ to $h$.
  \item Cancel what is left of flow paths in $g$ that traverse
    $\cut_i$ if they started at time $t_j\leq (i-\nu) \cdot \diam(G)$.
    Such flow paths correspond to packets that have been buffered
    (instead of forwarded) during many time steps.
  \end{inparaenum}
  Note that path flows of $f^*/2$ that start between cuts are added
  to $g$. The flow $h$ is a simple routing along shortest paths in which incoming
  flow (that needs to be further routed) is forwarded towards its
  destination without any buffering. (Note that $h$ does not use edges in $E_1$.)

  We claim that the choice of parameters implies that $h$ is a legal
  flow that succeeds in shipping all the flow that is transferred to
  it.  Consider all the flow that is transferred to $h$ in $\cut_i$.
  We show all this flow reaches its destination without any need to
  cross the next cut $\cut_{i+1}$.  Moreover, $h$ incurs a load of at
  most $1/2$ on each edge. Thus, the flow $h$ can be viewed as
  separate flows between consecutive cuts.

  Note that the time that elapses between two consecutive cuts equals
  the diameter of $G$.  This means that every flow path can be
  augmented by a shortest path to its destination before the next cut.

  To show that the load incurred by $h$ on each edge is at most $1/2$,
  suppose that all the flow that is transferred to $h$ in $\cut_i$
  traverses the same edge in $E_0$.  The amount of flow
  transferred to $h$ is bounded by $\alpha\cdot (\sum_{e\in E} c(e)+n\cdot B)\leq
  c_{\min}/2$, and hence the load in $h$ is bounded by $1/2$ as required.

  We claim that the throughput of $g+h$ is at least $(1-1/e)$ times
  the throughput of $f^*/2$. Indeed, flow is lost only when a residue
  of a flow path is canceled.  This happens only after a flow path
  traverses $\nu$ cuts.  By this time, the flow along this
  path has been decimated to a fraction of $(1-\alpha)^{\nu} \leq 1/e$ of its initial amount.

To complete the proof, note that the length of each flow path in $g+h$
is at most $(\nu+2)\cdot \diam(G)$. The number of edges of a flow path
in $g+h$ that are in $E_0$ is at most $\diam(G)$. The number of edges of a flow path in $g$
that are in $E_1$ is less than $(\nu+1)\cdot \diam(G)$, and flow paths in $h$ lack edges in $E_1$.
\end{proof}

\section{Proof sketch of Theorem~\ref{thm:algd}}\label{sec:proofdd}

The proof of Theorem~\ref{thm:algd} follows the proof of Theorem~\ref{thm:alg}.
The proof of the propositions below follows the analogous proofs in Section~\ref{sec:analysis}.

\begin{proposition}\label{prop:optd}
    $|f^*(R \mid \pmax)| \geq |\opt_f(R \mid \pmax)|$.
\end{proposition}

\begin{proposition}\label{prop:scaledd}
    $$\frac{1}{d+1}\cdot k^{d+1}\cdot (B+d\cdot c)\cdot |f^*_{\{1,d+1,\infty\}}(R \mid \pmax)| \geq |f^*(R \mid \pmax)| \geq |f^*_{\{1,d+1,\infty\}}(R \mid \pmax)|$$
\end{proposition}

\begin{proposition}\label{prop:fippd}
    $|\IPP| \geq \Omega \left(\frac{d+1}{k^{d+1}\cdot(B+d\cdot c)}\right) \cdot |f^*(R \mid \pmax)|$
\end{proposition}

\begin{proposition}\label{prop:preemptionsd}
    $|\IPP'| \geq \frac{1}{2k} \cdot |\IPP|$
\end{proposition}

\begin{proposition}\label{prop:lastd}\label{prop:Rsd}
    $|\alg| \geq \frac{1}{(d+1)\cdot k} \cdot |\IPP'|$
\end{proposition}

\noindent \textbf{Theorem~\ref{thm:algd}.} \emph{The competitive ratio of the
algorithm for $d$-dimensional grid networks is $$O\left(k^{d+3}
\cdot(B+d\cdot c) \right)=O\left(\log^{d+4} n \right)$$ provided that
$B,c \in [3,\log n]$.}
\medskip

\begin{proof sketch}{Theorem~\ref{thm:algd}}
    Bounding path lengths incurs a constant loss to the competitive
    ratio. Algorithm \IPP\ incurs an additional constant loss to the
    competitive ratio. The capacity assignment of $\{1,d+1\}$ reduces
    the throughput by a factor of $\frac{1}{d+1}\cdot k^{d+1} \cdot
    (B+d\cdot c)$.
    Similarly to the uni-dimensional case, a fraction of at most $(1-\frac{1}{2k})$ of the requests in $\IPP$ are preempted before they reach their last cube.
    Finally, a fraction of at least $\frac{1}{(d+1)\cdot k}$ of the requests that reach their last tile are successfully routed, i.e.,  by detailed routing in the last tile.
    Hence, the total fraction of requests that are successful routed is $\Omega \left(\frac{1}{k^{d+3}\cdot(B+d\cdot c)}\right)$.
    The theorem follows since $B,c \in [3,\log
    n]$.
\end{proof sketch}

\section{Proof of Theorem~\ref{thm:bufferless}}\label{sec:proofs}

\paragraph{Theorem~\ref{thm:bufferless}.}
  \emph{There exists an online deterministic preemptive algorithm for packet
  routing in bufferless $d$-dimensional grids with a competitive ratio of $O(\log ^{d+2}
  n)$.}
\medskip

\begin{proof}
    Since $B=0$, the space-time graph $G^{st}$ after untilting
    consists of unconnected $d$-dimensional grids.
    Within each such $d$-dimensional grid, we apply a version of our algorithm.  Note that since $B=0$, trivially $\pmax \leq
    \sum_i\ell_i$ (i.e., the diameter of the grid) and does not depend
    on $c$. Note also that the destination is a single node
    $(b_i,t')$, where $t'=t_i+\|a_i-b_i\|_1$. Thus we need not
    introduce sink nodes. The edge capacities are $d \cdot c$ to every interior edge (instead of
    $(d+1)$). Hence, the capacity assignment reduces the throughput by
    a factor of $k^{d}$ (instead of $k^{d+1} \cdot
    (B+d\cdot c)$).
  \end{proof}

\section{Proof of Lemma~\ref{lemma:revMarkov}}\label{sec:RevMarkovproof}
\begin{lemma}[\protect{\textbf{A Reverse Markov Inequality}}]\label{lemma:revMarkov}
  Let $X$ be a nonnegative bounded random variable attaining values in $[0,a]$. For every $d < a$,
  \begin{eqnarray*}
    \Pr \left(X \geq d \right) & \geq & \frac {E(X)- d}{a- d}\:.
  \end{eqnarray*}
\end{lemma}

\begin{proof}
    We prove that $\Pr \left(X <  d \right)  \leq  1-\frac {E(X)- d}{a- d}$. Let Y be a random variable such that $Y \triangleq a - X$. Note that, $Y$ is also a nonnegative bounded random variable attaining values in $[0,a]$. Hence, $X < d$ if and only if $Y > a-d$. The expected value of $Y$ is $E(Y)=a-E(x)$. The lemma follows by applying Markov Inequality~\cite{MU}, as follows:
    \begin{eqnarray*}
        \Pr \left(X < d \right) & = & \Pr \left(Y > a - d \right) \\
        & \leq & \frac {E(Y)}{a - d}\\
        & = & \frac {a-E(x)}{a - d}\\
        & = & 1- \frac {E(x)- d}{a - d}\:.
  \end{eqnarray*}
\end{proof}

\section{Online Integral Path Packing Algorithm \route}
\label{sect:routealg} \label{sec:IPP} In this section we present algorithm \route\
and prove Theorem~\ref{thm:IPP}. The presentation follows the framework
of~\cite{BN06,BNsurvey}. The presentation emphasizes two points: (1)~The graph over
which the requests arrive may be infinite. (2)~There is an upper bound $\pmax$ on the length
of a path that may serve a request.

\paragraph{Linear Programming Formulation.}
Fractional path packing is a multi-commodity flow problem, and is formulated by a
linear program (LP). In Figure~\ref{fig:LP}, the dual LP corresponds to the
fractional path packing problem as well as the corresponding primal LP are listed.

The notation in the LPs is as follows. For each request $i$, let $P_i$ denote the set
of paths in $G$ that can serve the request $r_i=(a_i,b_i)$. The length of every path
$p \in P_i$ is at most $\pmax$.  The variables $f(i,p)$ denote the amount of flow
allocated to request $i$ along the path $p$. The demand constraint in the dual LP
states that at most one unit of flow can be jointly allocated to all the paths in
$P_i$. The capacity constraint states that at most $c(e)$ units of flow can traverse
an edge $e$. The objective is to maximize the flow amount.

The primal LP has two types of variables: one variable $z_i$ per request $r_i$ and
one variable $x_e$ per edge $e$. The variable $x_e$ can be interpreted as a weight
assigned to the edge $e$. The covering constraint states that for every request $r_i$
and every path $p\in P_i$, the weight of the path $p$ plus $z_i$ should be at least
$1$. The objective is to minimize the sum of edge weights times their capacities plus
the sum of the variables $z_i$.

\begin{figure}
\centering
  \begin{tabular}{c}
  \centerline{\fbox{\begin{minipage}{0.7\textwidth}
 \begin{center}
        \begin{eqnarray*}
        \min\sum_{e\in E}x_{e}\cdot c(e)+\sum_{i}z_i ~~~s.t. &&\\
\forall i~\forall p\in P_{i} ~: ~~~ \sum_{e\in p}x_{e}+z_{i} && \geq 1
~~~\text{(covering const.)}
\\
x,z &&\geq 0
\end{eqnarray*}
 \end{center}
    \end{minipage}}}
    \\ (I) \\
    \centerline{\fbox{\begin{minipage}{0.7\textwidth}
\begin{center}
        \begin{eqnarray*}\max \sum_{i}\sum_{p\in P_{i}} f(i,p) & s.t.\\
        \text{(demand const.)}&\forall i  & \sum_{p\in P_{i}}f(i,p)\leq 1\\
        \text{(capacity const.)}& \forall e\in E &  \flow(e) \leq c(e)\\
&&f\geq 0
\end{eqnarray*}
\end{center}
    \end{minipage}}}
 \\ (II) \\
  \end{tabular}
  \caption{
(I) The Primal linear program.
    (II) The Dual linear program.}
   \label{fig:LP}
\end{figure}

\paragraph{The Online Algorithm for Integral Packing of Paths.}
The listing of algorithm \route\ appears in
Figure~\ref{fig:route}.  Note that the graph $G=(V,E)$ may
be infinite.  This implies that the primal LP has an
infinite number of variables (however, all but a finite
subset of the primal LP variables are zero).  We assume
that there exists a lightest path oracle that, given edge
weights $x_e$ and a request $r_i$, finds a lightest path $p
\in P_i$.

\begin{algorithm}
    Input: $G=(V,E)$ (possibly infinite), sequence of requests $\{r_i\}_{i=1}^\infty$ where $r_i\triangleq (a_i,b_i)$.
    \\
    \textbf{Upon arrival} of request $r_i$:
        \begin{enumerate}
            \item Let $\alpha(p,i) \triangleq \sum_{e\in p} x_{e}$.
            \item $p \leftarrow \textrm{argmin} \{\alpha(p',i) : p'\in
                  P_i\: \}$ (find a lightest path from $a_i$ to $b_i$ using an oracle).
            \item If $\alpha(p,i) < 1$ then, \textbf{route} $r_i$ along $p$:
            \begin{enumerate}
                    \item $f(i,p) \leftarrow 1$.
                    \item \label{step:xupdate}For each $e\in p$ do
                  \begin{align*}
                        x_{e} \gets&
                          x_{e} \cdot 2^{1/c(e)}+\frac{1}{\pmax}\cdot (2^{1/c(e)}-1)\:.
                  \end{align*}
                  \item $z_{i} \leftarrow 1 - \alpha(p,i)$.
            \end{enumerate}
            \item Else, \textbf{reject} $r_i$.
            \begin{enumerate}
              \item$z_{i} \leftarrow 0$.
            \end{enumerate}
        \end{enumerate}
\caption{The \route\ algorithm.  We assume that all the
  variables are initialized to zero using lazy initialization. We assume that given edge variables $x_e$, there exist an oracle that returns a lightest path in $P_i$. }
  \label{fig:route}
\end{algorithm}

For a given sequence $\sigma$ of requests let $F^*(\sigma)$
denote the maximum flow of the dual LP.  An online integral
path packing algorithm is said to be
\emph{$(\alpha,\beta)$-competitive} if for every sequence
$\sigma$ of requests (1)~its total throughput is at least
$F^*(\sigma)/\alpha$, and (2)~the load of every edge is at
most $\beta$.

\medskip
\noindent
The proof of the following theorem follows the framework of~\cite{BN09,BNsurvey}.

\paragraph{Theorem~\ref{thm:IPP}.}
\emph{ Algorithm \route\ is a $(2,\log(1+ 3\cdot \pmax))$-competitive online integral
  path packing algorithm under the following assumptions: (1)~ $\min_{e} c(e) \geq
  1$.  (2)~A path is legal if it contains at most $\pmax$ edges.
(3)~There is an oracle, that given edge weights and a
  request, finds a
  lightest legal path from the source to the destination.
}

\begin{proof}
  Let us denote by $\Delta_i P$ (respectively, $\Delta_i D$) the change in the primal
  (respectively, dual) cost function after request $r_i$ is processed.
We claim that $\Delta_i P\leq 2\cdot \Delta_i D$.

If $r_i$ is rejected, then $\Delta_i P=\Delta_i D=0$.  If $r_i$ is accepted, then
$\Delta_i D = 1$ and $\Delta_i P = \sum_{e \in p} \Delta_i x_{e}\cdot c(e) + \Delta_i
z_{i}$.  Step (3b) increases the cost $\sum_{e} x_{e}\cdot c(e)$ as follows:
    \begin{eqnarray}
    \label{eq:deltax}
    \sum_{e} \Delta_i x_{e}\cdot c(e)
    & = &
    \sum_{e \in p}
    \left[x_{e} \cdot( 2^{1/c(e)}-1)+\frac{1}{\pmax}\cdot (2^{1/c(e)}-1)\right]
    \cdot c(e)  \nonumber\\
    & = &
    \sum_{e\in p} \left(x_{e}+\frac{1}{\pmax}\right) \cdot
    ( 2^{1/c(e)}-1)\cdot c(e)\nonumber\\
    & \leq &
    c_{\min} \cdot (2^{1/c_{\min}}-1)
    \sum_{e\in p} \left(x_{e}+\frac{1}{\pmax}\right)\nonumber\\
    & \leq &
    1 \cdot (2^{1}-1)\sum_{e\in p} \left(x_{e}+\frac{1}{\pmax}\right) \nonumber\\
    & \leq &
    \sum_{e\in p} x_{e}+\sum_{e\in p} \frac{1}{\pmax}\nonumber\\
    & \leq & \alpha(p,i)+1\:.
\end{eqnarray}

Hence after step (3c):

\begin{eqnarray}
    \Delta_i P
    & = & \sum_{e \in p} \Delta_i x_{e}\cdot c(e) + \Delta_i  z_{i}\nonumber\\
    & \leq & (\alpha(p,i) +1) + (1- \alpha(p,i))\nonumber\\
    & = & 2\:.
\end{eqnarray}
Since $\Delta_i D = 1$ it follows that $\Delta_i P \leq 2 \cdot \Delta_i D$, as
required.

After dealing with each request, the primal variables $\{x_{e}\}_{e} \cup
\{z_{i}\}_{i}$ constitute a feasible primal solution. Given a dual solution
$\{f(i,p)\}$, let $|f|\triangleq\sum_{i}\sum_{p\in P_{i}} f(i,p)$. Let $\{f^*(i,p)\}$
denote an optimal dual solution.  Using weak duality and since $\Delta_i P \leq
2\cdot \Delta_i D$ it follows that:
        \begin{eqnarray}
        |f^*| \leq & \sum_{e\in E} x_{e}\cdot c(e)+\sum_{i}z_{i} \leq & 2\cdot|f|\:,
        \end{eqnarray}
    which proves $2$-competitiveness; namely $|f|\geq \frac{1}{2}\cdot|f^*|$.

    We now prove $\log (1+3 \cdot \pmax)$-feasibility of the dual solution, i.e. for each $e\in E$, $\flow(e) \leq \log (1+3 \cdot \pmax)$.
    The update rule of the primal variables $\{x_e\}_e$ in Step~\ref{step:xupdate} implies,
     \begin{eqnarray}
\nonumber        x_{e} & = & \frac{1}{\pmax}(2^{1/c(e)}-1)\cdot \sum_{j=0}^{\flow(e)-1}(2^{1/c(e)})^j \\
\nonumber              & = & \frac{1}{\pmax}(2^{1/c(e)}-1)\cdot
              \frac{2^{\flow(e)/c(e)}-1}{2^{1/c(e)}-1}\\
\label{eq:xe}              & = & \frac{2^{\flow(e)/c(e)}-1}{\pmax}\:.
     \end{eqnarray} 
     The update rule requires that $\alpha(p,i)<1$ for every $p$.  Hence, before the
     update $x_e<1$, and after the update $x_e < 2^{1/c(e)}+\frac{1}{\pmax}\cdot
     (2^{1/c(e)}-1)$.  Since $c_{\min} \geq 1$, it follows that $x_e<3$.

By Equation~\eqref{eq:xe} it follows that
    \begin{eqnarray*}
    \frac{2^{\flow(e)/c(e)}-1}{\pmax}&<& 3\:.
    \end{eqnarray*}
    Implying that $\flow(e) \leq \log (1+3 \cdot \pmax)\cdot c(e)$, as required.
\end{proof}

\section{Two Models For Nodes in Store-and-Forward Networks}\label{sec:model}
\newcommand{\comb}{\text{\emph{comb}}} The literature contains two
different models of node functionality. In an effort to make the
comparison concrete and perhaps clearer, we present schematic
implementations of the nodes in each model.

To simplify the discussion, we use two type of packets: regular packets and ghost
packets. A regular packet contributes a unit to the throughput (if delivered) and a
ghost packet does not contribute to the throughput and acts as a ``place holder''.
We therefore may treat a buffer as if it always contains $B$ packets.
If a buffer contains only ghost packets, then it is empty in reality.
A reasonable policy does not drop a regular packet while keeping a ghost packet.

\paragraph{Model 1.}
This model is used by~\cite{ARSU,RR}. Figure~\ref{fig:node1}
depicts a block diagram of a node. A node contains a combinational
circuit \comb, a buffer consisting of $B$ flip-flops, and $c$
flip-flops on each  link that emanates the node.

In each clock cycle, the combinational circuit $\comb$ receives $c$
packets from each incoming link, $B$ packets from its buffer, and
$B+c$ packets from its local inputs. It outputs $B$ packets to the
buffer and $c$ packets along each outgoing link. Packets that were
input but not output are considered dropped packets unless the node is
their destination.

\paragraph{Model 2.}
This model is used by~\cite{AKK,AZ}. Figure~\ref{fig:node2} depicts a
block diagram of a node. A node contains two combinational circuits
$\comb_0$ and $\comb_1$, two sets of $B$ latches, and one latch
on the link that emanates the node. Note that this implementation uses a
two-phase clock. The phases are denoted by $\phi_0$ and $\phi_1$.

In the first phase of each clock cycle, the combinational circuit $\comb_0$ receives
one packet from the incoming link, $B$ packets from its buffer, and $B$ packets from its
local input. In total $2B+1$ packets (either regular or ghost packets) are fed to the
$\comb_0$ circuit. The $\comb_0$ circuit outputs $B$ packets and the rest are dropped
unless this is their destination.  In the
second clock phase of each clock cycle, the combinational circuit $\comb_1$ outputs
one packet along the outgoing link and $B$ packets are sent back to $\comb_0$.

\paragraph{Remarks:}
\begin{enumerate}
\item The setting $B=c=1$ in Model 1 is strictly stronger than $B=1$ in Model 2.
  Indeed, in Model 1, if a node receives a regular packet from its neighbor and is
  also input a regular packet locally, then it may store one packet and forward the
  other one.  On the other hand, in Model 2, one of the packets must be dropped.

\item We could also allow for more injected packets in each node. In
  this case, the node must drop some of them. Of course, the online
  algorithm has to decide which packets should be dropped.
\item The linear lower bounds for $B=1$ in~\cite{AZ,AKK} hold only with
  respect to Model 2.
\item It is not clear how to extend Model 2 for the case that $c>1$ or $B=0$.
\item Under the common assumption that the cost of a flip-flop is
  roughly twice the cost of a latch, the hardware needed for the
  latches of a node in Model 2 is roughly the same as the cost of
  flip-flops of a node in Model 1 (with $c=1$).
\end{enumerate}

\begin{figure}[h!]%
  \centering
  \subfloat[]{
\includegraphics[width=0.9\textwidth]{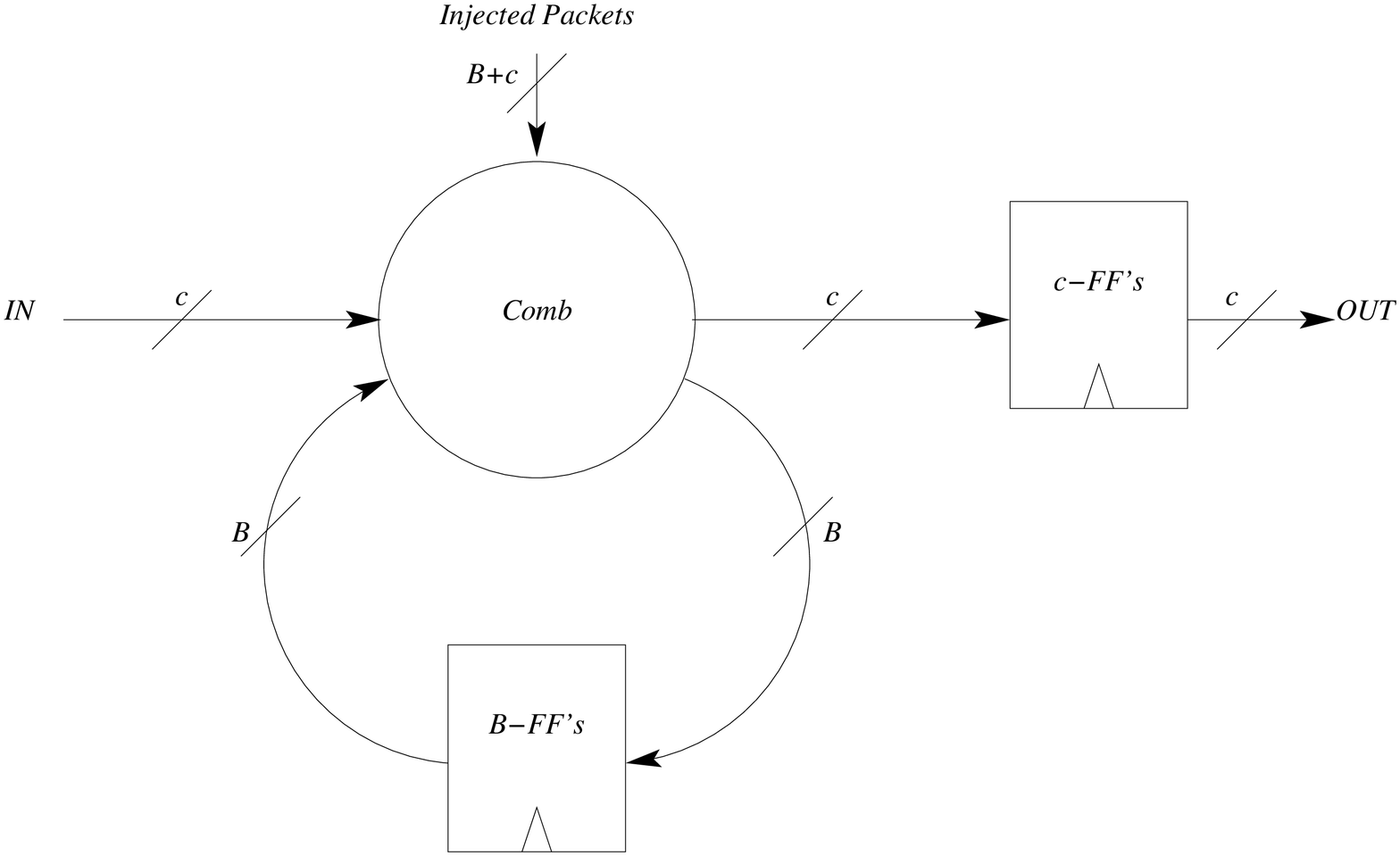}
\label{fig:node1}
}%
\\
  \subfloat[]{
\includegraphics[width=0.9\textwidth]{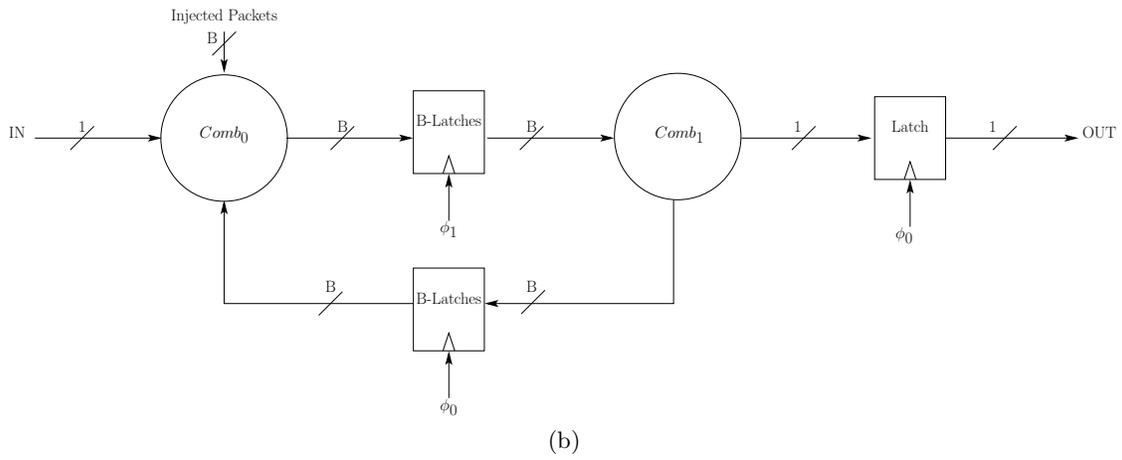}
\label{fig:node2}
}%
%
\caption{(a) A schematic of a node in Model-1. (b) A schematic of a node in Model-2.}
\end{figure}
\end{document}